\documentclass{article}

\usepackage{arxiv}

\usepackage[utf8]{inputenc} 
\usepackage[T1]{fontenc}    
\usepackage{hyperref}       
\usepackage{url}            
\usepackage{booktabs}       
\usepackage{amsfonts}       
\usepackage{nicefrac}       
\usepackage{microtype}      
\usepackage{lipsum}		
\usepackage{graphicx}
\usepackage[round]{natbib}
\usepackage{doi}

\title{Constrained Markov decision processes for response-adaptive procedures in clinical trials with binary outcomes}

\author{Stef Baas \\
	Stochastic Operations Research Group,
	University~of~Twente,
	7522 NH Enschede, The~Netherlands\\
		\And 
	Aleida Braaksma \\
	Stochastic Operations Research Group,
	University~of~Twente,
	7522 NH Enschede, The~Netherlands\thanks{Corresponding author}
	\And
	Richard J. Boucherie\\
	Stochastic Operations Research Group,
	University~of~Twente,
	7522 NH Enschede, The~Netherlands\\
}

\renewcommand{\shorttitle}{CMDP for RA procedures in clinical trials with binary outcomes}

\hypersetup{
pdftitle={A template for the arxiv style},
pdfsubject={q-bio.NC, q-bio.QM},
pdfauthor={David S.~Hippocampus, Elias D.~Striatum},
pdfkeywords={First keyword, Second keyword, More},
}
\usepackage{bm}
\def\mydefb#1{\expandafter\def\csname b#1\endcsname{\bm{#1}}}
\def\mydefallb#1{\ifx#1\mydefallb\else\mydefb#1\expandafter\mydefallb\fi}
\mydefallb ABCDEFGHIJKLMNOPQRSTUVWXYZabcdeghijklnopqrstuyz\mydefallb

\def\mydefgreek#1{\expandafter\def\csname b#1\endcsname{\text{\boldmath$\mathbf{\csname #1\endcsname}$}}}
\def\mydefallgreek#1{\ifx\mydefallgreek#1\else\mydefgreek{#1}%
	\lowercase{\mydefgreek{#1}}\expandafter\mydefallgreek\fi}
\mydefallgreek {beta}{Gamma}{alpha}{Delta}{epsilon}{Omega}{Theta}{Iota}{Lambda}{kappa}{mu}{nu}{Xi}{Pi}{chi}{rho}{Sigma}{Psi}{tau}\mydefallgreek
\newcommand{\bv}{\bm v}
\newcommand{\bw}{\bm w}
\newcommand{\bx}{\bm x}

\newcommand{\Tau}{\mathrm{T}}
\usepackage{lipsum}

\let\oldcite\cite
\renewcommand{\cite}[2][]{\mbox{\oldcite[#1]{#2}}}
\let\oldcitep\citep
\renewcommand{\citep}[2][]{\mbox{\oldcitep[#1]{#2}}}
\let\oldcitet\citet
\renewcommand{\citet}[2][]{\mbox{\oldcitet[#1]{#2}}}
\newcommand{\bff}{\bm f}

\newcommand{\C}{\text{\normalfont C}}
\newcommand{\D}{\text{\normalfont D}}

\usepackage{graphicx}%
\usepackage{multirow}%
\usepackage{amsmath,amssymb,amsfonts}%
\usepackage{amsthm}%
\usepackage{mathrsfs}%
\usepackage[title]{appendix}%
\usepackage{xcolor}%
\usepackage{textcomp}%
\usepackage{manyfoot}%
\usepackage{booktabs}%
\usepackage{algorithm}%
\usepackage{algorithmicx}%
\usepackage{algpseudocode}%
\usepackage{listings}%
\algnewcommand{\Inputs}[1]{%
	\State \textbf{Inputs:}
	\Statex \hspace*{\algorithmicindent}\parbox[t]{.8\linewidth}{\raggedright #1}
}
\algnewcommand{\Initialize}[1]{%
	\State \textbf{Initialize:}
	\Statex \hspace*{\algorithmicindent}\parbox[t]{.8\linewidth}{\raggedright #1}
}
\usepackage{placeins}
\usepackage{accents}
\newcommand{\ubar}[1]{\underaccent{\bar}{#1}}
\usepackage{float}

\usepackage{orcidlink}
\usepackage{hyperref}                                           

\hypersetup{
bookmarksnumbered,
pdfstartview={FitH},
linkcolor=blue,
citecolor=blue,
colorlinks=true,
pdfpagemode={UseNone},
plainpages=false
}%

\newtheorem{theorem}{Theorem}
\newtheorem{lemma}[theorem]{Lemma}

\newtheorem{example}[theorem]{Example}
\begin{document}
\maketitle
\begin{abstract} A \emph{constrained Markov decision process}~(CMDP) approach is developed for response-adaptive procedures in clinical trials with binary outcomes. The resulting CMDP class of Bayesian response-adaptive procedures can be used to target a certain objective, e.g., patient~benefit or power while using constraints to keep other operating characteristics under control. In the CMDP approach, the constraints can be formulated under different priors, which can induce a certain behaviour of the policy under a given statistical hypothesis, or given that the parameters lie in a specific part of the parameter space. A solution method is developed to find the optimal policy, as well as a more efficient method, based on backward recursion, which often yields a near-optimal solution with an available optimality gap. Three applications are considered, involving type I error and power constraints, constraints on the mean squared error, and a constraint on prior robustness. While the CMDP approach slightly outperforms the \emph{constrained randomized dynamic programming}~(CRDP) procedure known from literature when focussing on type~I and II error and mean squared error, showing the general quality of CRDP, CMDP significantly outperforms CRDP  when the focus is on type I and II error only.
\end{abstract}




\maketitle
\newpage
\section{Introduction}
 The current gold standard for assessing the efficacy of an experimental treatment is the \emph{randomized controlled trial}~(RCT)~\citep{bhatt2010evolution}. Historically,~participants enrolled in an RCT are randomized to each treatment, where the probability of allocation is independent of the history of outcomes and allocations in the trial (the trial history). This approach induces a high quality of statistical \emph{operating characteristics} (OCs) such as the variance of the treatment effect estimator or power of statistical tests, making it possible to make statements on how the experimental treatment will perform for a new patient, not considered in the trial (extrapolation), which is important in settings with large patient populations~\citep{Palmer219}. Furthermore, independence of the trial history and allocations ensures comparability of treatment groups and robustness to time trends. 
A downside of this approach is that a large part of the patients included in the trial might obtain inferior treatment. For a two-arm trial with equal allocation probabilities, this could concern half the patients, while for a multi-arm trial, this percentage can be even higher. The effect of this depends on the severity and nature of the illness. For instance, the effect is greater when the illness is very severe. Also, for rare diseases, a large part of the patients having the disease will be included in the trial, and in this setting extrapolation of trial results might be considered less important. In such cases, it could be desirable to make a different trade-off between individual ethics and collective ethics~\citep{Heilig2005critical}. 

An alternative, aiming to make such a trade-off, is a \emph{response-adaptive}~(RA) procedure~(see, e.g., \citet{rosenberger1996directions}, \citet{antognini2015adaptive}, or \citet{Villar2022}).
The general idea behind an RA procedure, introduced in~\citet{thompson1933likelihood}, is to sequentially base the allocation of the next arriving participant to a treatment on the current trial history, aiming to reach a certain objective, e.g., high patient~benefit or good statistical OCs. 
RA~procedures have been proposed where the allocation probability is strictly between zero and one for all trial histories, or where the allocation is deterministic for at least one trial history. An RA~procedure in the former class is termed a \emph{response-adaptive randomization}~(RAR) procedure, while an RA~procedure in the latter class is termed a \emph{deterministic response-adaptive}~(DRA) procedure. A clinical trial using an RA procedure will from now on be termed an RA design, while a design where treatment and outcomes are independent will be termed a \emph{non-adaptive}~(NA) design.
In recent years, RA procedures have found a growing number of applications in clinical trials~\citep{viele_berry_comment}, where most of these applications involved RAR procedures.
The US FDA has encouraged to consider multiple possibilities for statistical trial design, including RAR procedures in~\citet{FDA2010guidance}, while this guidance document also highlights the controversies of RA procedures (for an overview and comment, see~\citet{robertson2023response}).

The current paper focuses on optimisation-based \emph{Bayesian RA} (BRA) procedures.  Optimisation-based BRA procedures focus on the exploration-exploitation trade-off, which, in a clinical trial, corresponds to the aim of learning which treatment is best while simultaneously allocating the highest amount of patients to the best treatment. 
For a recent literature review including optimisation-based BRA approaches see \citet[Chapter 2]{williamson2020bayesian}.
In the literature, BRA procedures can mainly be classified into three categories: 
\begin{itemize}
    \item{\bf Index-based approaches}\\
    For an index-based BRA procedure, an index value is independently determined based on the data collected for each treatment group, and the treatment with the highest index value has the highest probability of being allocated to the next patient by the RA procedure.  Examples are DRA approaches such as the Gittins index procedure~\citep{Gittins1979bandit}, Bayes-UCB procedure~\citep{kauffmannucb}, and RAR procedures such as semi-randomized index-based approaches~\citep{bather1981Randomized}.\\
\item{\bf Thompson sampling (and modifications)}\\
    For Thompson sampling (introduced in~\citet{thompson1933likelihood}) the treatment with the highest posterior probability of having the highest expected outcome has the highest probability of being allocated to the next patient. A modification of Thompson sampling has been proposed in~\citet{THALL2007859} which induces a smaller variance in the allocation probabilities. \\
       \item{\bf Markov decision process (MDP) approaches}\\
 MDP approaches, introduced in~\citet{Bradt1956}, allocate treatment to trial~participants with the aim to maximize the patient~benefit up to a fixed horizon, given a prior distribution on the parameters for the model.  The MDP  approach will be the main focus of the current paper. 
\end{itemize}
Out of the three optimisation-based BRA procedures, the MDP approach is the most computationally intensive, as shown in~\citet{villar2015MAB}.   
As (modifications of) Thompson sampling can directly be applied in blocked randomized designs, the second approach is currently the most common approach when implementing RA procedures in practice~\citep{viele_berry_comment}, while the other two approaches have the potential to show the highest patient~benefit (see, e.g., \citet{villar2015MAB} or \citet{williamson2017bayesian})

The MDP approach is a natural optimisation method for the setting of a clinical trial, as every outcome is weighted equally and a finite trial horizon is taken into account~\citep{hardwick1995modified}. A variety of different MDP BRA procedures have been proposed in the literature, showing the flexibility of this approach. In~\citet{berry1995adaptive}, an MDP BRA procedure was introduced which optimises, under a Bayesian model, the expected outcomes of~participants in the trial as well as the expected outcomes of a finite number of patients which are all allocated to one treatment after the trial is completed. This procedure hence makes an explicit trade-off between individual and collective ethics. 
While RAR procedures can be constructed based on index-based BRA procedures~\citep{bather1981Randomized}, the randomization component is not taken into account in the optimisation.  
\citet{cheng2007optimal} introduced an MDP Bayesian RAR (BRAR) procedure where the randomization was included in the optimisation in a natural way. \citet{williamson2022generalisations} considered extensions of this approach, including constraints on the minimum allocations to both treatment groups, delays in outcomes, and random arrivals. 
\citet{ hardwick2008response} introduced a multi-objective MDP BRA procedure, where the two objectives revolved around patient~benefit and best treatment selection.
 In~\citet{MERRELL2023107599} an MDP BRA procedure was formulated for block-wise allocation, optimising the block sizes, number of blocks, and treatment allocations for blocks, where in the objective a trade-off is made between the expected number of successes and OCs.
\citet{YI2023125} formulated an infinite-horizon MDP for allocation in a clinical trial with general outcomes and showed that the allocation ratio converges, ensuring the validity of likelihood-based tests. 

The current paper introduces a novel class of constrained MDP BRA procedures,~CMDP procedures for short. The CMDP  procedures follow from a modified version of a constrained Markov decision process~\citep{altman1999constrained}, where the modification is that it is possible to use different expectation operators in each constraint. This modification is needed to optimise the allocation of treatment while keeping, e.g., type~I~error and power under control. In MDP BRA procedures, optimisation is often performed under a vague prior, while the probability distribution of the outcomes under, for instance, a type~I~error constraint can have strong assumptions on the parameters such as equality of the expectation of the outcomes, hence different prior distributions are needed to formulate Bayesian type~I and power constraints. A CMDP procedure can be made to satisfy requirements on OCs for the trial using constraints while optimising patient~benefit.
In comparison to previous methods from the literature, 
the CMDP approach can be considered a more natural way to directly impose a desired behaviour for an RA procedure while ensuring that, given the imposed constraints, the obtained policy is optimal, which makes it possible to reach a higher patient~benefit.
An additional advantage over penalized methods is that adding constraints will have an effect that is known beforehand, namely it will shrink the feasible region for the policies, whereas the effect of adding a penalty term to the objective is less clear beforehand. 

Throughout the current paper we make use of an efficient implementation of backward and forward recursion following~\citet{jacko2019binarybandit}, where we make use of a conservation law for the states, use a storage mapping function to store values efficiently, and overwrite elements of the value function not used further in the algorithm.
This makes it possible to compute the MDP BRA procedures considered in this paper in a relatively short amount of time. Furthermore, this makes it possible to compute the values of clinical trial OCs directly instead of approximating them by simulation, avoiding Monte Carlo error. 

The current paper is structured as follows. Section~\ref{sect:model} introduces the model of a binary two-arm clinical trial where outcomes are collected using a  response-adaptive procedure, as well as relevant operating characteristics. Section~\ref{sect:CMDP} introduces the class of CMDP procedures, as well as an algorithmic method for determining CMDP procedures. In section~\ref{sect:applications} CMDP procedures are constructed that optimise patient~benefit under a restriction on the power and type~I~error, a restriction on the mean squared error, and also on a restriction on robustness to prior misspecification. Section~\ref{sect:discussion} concludes the paper and gives directions for future research.

\section{Model and operating characteristics}\label{sect:model}
\subsection{Two-arm Response adaptive design with binary outcomes}~\label{subsect:model}
 We consider a trial, in which there are two treatments (arms) with unknown outcome distributions, the control and developmental treatment.
Trial~participants are sequentially allocated to a treatment using a response-adaptive (RA) procedure which, given the current trial history, determines the probability that the next participant obtains a given treatment. After allocation, the (binary) outcome of the participant, sampled from a Bernoulli distribution, becomes available and is added to the trial history before allocating the next participant.

We now make the above formal. Let~$\btheta=(\theta_\C,\theta_\D)\in[0,1]^2$ 
be a tuple of (unknown) success probabilities, where~$\C$ denotes the control treatment and~$\D$ 
denotes the developmental treatment. In the following, the same convention (i.e., first~$\C$ then~$\D$) will be used to construct tuples from variables for the control and developmental treatment.
For a fixed trial size~$n\in\mathbb{N}$, let~$\bY = (Y_{a,t})_{a\in\{\C,\D\},t\in\{1,\dots, n\}}$ be a sequence of independent Bernoulli random variables, where~$\mathbb{P}_\btheta(Y_{a,t}=1)=\theta_a$ for~$a\in\{\C,\D\}$. The random variable~$Y_{a,t}$ denotes a potential outcome for trial participant~$t$ under treatment~$a$.
Let~$\mathcal{H}=\bigcup_{t=0}^\infty\mathcal{H}_t$ where~$\mathcal{H}_0=\{()\}$ only contains the empty tuple and~$
     \mathcal{H}_t=\{(a_1,y_{1},a_2, y_{2},\dots, a_t,y_{t}): y_{w}\in\{0,1\},\,a_w\in\{\C,\D\},\;\forall w\leq t\} \label{defn:histories}
$
is the set of possible trial histories of outcomes and actions up to participant~$t$.
An RA procedure~$\pi:\mathcal{H}\mapsto [0,1]$ maps a trial history to the probability that the subsequent participant is allocated to the control treatment. Trial~participants are allocated sequentially, after which the outcome for that trial participant is observed before allocating the next participant, i.e.,  letting~$\bH_0=()$ we recursively define the realised trial history as~$\bH_t=(A_1,Y_{A_1,1},A_2,Y_{A_2,2},\dots, A_t, Y_{A_t,t})$, where each~$A_{t}$ is an independently drawn Bernoulli random variable with~$\mathbb{P}(A_t=1)=\pi(\bH_{t-1})$ for~$t=1,\dots, n$.

Let~$S_{a,t}$ be the recorded number of successes, and~$N_{a,t}$ denote the number of allocations for arm~$a$ up to time~$t$, i.e.,~$$S_{a,t}= \sum_{u=1}^t Y_{a,u}\mathbb{I}(A_u=a),\indent N_{a,t} = \sum_{u=1}^t \mathbb{I}(A_u = a),\indent \forall a\in\{\C,\D\}, \;t\in\{1,\dots,n\},~$$ where~$\mathbb{I}$ denotes the indicator function. 
Let~$\bX_t = (\bS_{t},\bN_{t})$ be a state variable containing the successes and allocations for each arm up to time~$t$. Letting~$\bX_0=((0,0),(0,0))$, the state space for~$\bX=(\bX_t)_{t=0}^n$ is~$\mathcal{X} = \cup_{t=0}^n\mathcal{X}_t$ where for all~$t$:$$ \mathcal{X}_t=\{((x_{11},x_{12}),(x_{21},x_{22})): x_{ij}\in \{0,\dots, t\},\,x_{1j}\leq x_{2j},\;x_{21}+x_{22}=t,\;\forall i,j\in\{1,2\} \}.$$
Let~$\partial\bs_\C = ((1,0),(1,0))$ and~$\partial\bff_\C = ((0,0),(1,0))$ be the change in~$\bX_t$ after a success and failure for the control arm, and let~$\partial \bs_\D$,~$\partial\bff_D$ be defined similarly. Letting~$q(\bx_{t+1},\bx_t,a) = \mathbb{P}_\btheta(\bX_{t+1} = \bx_{t+1}\mid\bX_t=\bx_t,A_t = a)$, we have for \mbox{ all~$t\in\{0,\dots,n\}$,~$\bx_t\in\mathcal{X}_t,\bx_{t+1}\in\mathcal{X}_{t+1}$,~$a\in\{\C,\D\}$, and~$\btheta\in[0,1]^2$}

\begin{equation}
q(\bx_{t+1},\bx_t,a)= \begin{cases}
\theta_a,\indent &\text{ if~$\bx_{t+1} = \bx_t +  \partial \bs_a$,}\\
(1-\theta_a),\indent &\text{ if~$\bx_{t+1} = \bx_t +\partial \bff_a$,}
\end{cases}\label{Transdyd_state}\end{equation}
where the addition of tuples is understood to be element-wise.
If the RA procedure~$\pi$ is a function of~$\bX_t$, i.e, can be written as a function~$\pi:\mathcal{X}\mapsto [0,1]$, the process~$\bX= (\bX_t)_t$ is a Markov process with  transition structure 
$$ \mathbb{P}^\pi_\btheta(\bX_{t+1}=\bx_{t+1}\mid \bX_t=\bx_t) = \sum_{a\in\{\C,\D\}} q(\bx_{t+1},\bx_t,a)\pi(\bx_t)^{\iota_a}(1-\pi(\bx_t))^{1-\iota_a}$$
for all~$\bx_t\in\mathcal{X}_t,\bx_{t+1}\in\mathcal{X}_{t+1}$, where~$\mathbb{P}^{\pi}_\btheta$ denotes the probability measure on states induced by the RA procedure~$\pi$ and~\eqref{Transdyd_state} and~$\iota_a=\mathbb{I}(a=\C)$ for all~$a\in\{\C,\D\}$. In this case,~$\pi$ is a Markov RA procedure~\citep{yi2013exact}.

\subsection{Definition of operating characteristics} \label{sect:ocs}
In this section, we define important OCs that can be calculated based on the model introduced at the start of Section~\ref{sect:model}. 

First, at decision epoch~$t=n$, denoted the trial horizon, we test for a treatment effect~$\theta_\D-\theta_\C$, i.e., we test
   ~$$H_0:\theta_\D-\theta_\C=0\indent \text{v.s.}\indent H_1:\theta_\D-\theta_\C\neq 0.$$
    The test of choice is often Fisher's exact test which, for a significance level~$\alpha\in(0,1)$, rejects when~$\Tau(\bX_n)\leq \alpha$, where, letting~${s}_a(\bx_t)$ and~${n}_a(\bx_t)$ be the number of successes and allocations encoded in~$\bx_t$ and~$s(\bx) = s_\C(\bx)+s_\D(\bx)$ for all~$\bx\in\mathcal{X}$,
   ~$$\indent \Tau(\bX_n)= \sum_{\substack{\bx'_n\in\mathcal{X}_n\\
    P(\bx'_n)\leq P(\bx_n)}} P(\bx'_n),\indent \text{ and }\indent  P(\bx_n)= \frac{\binom{n_\C(\bx_n)}{s_\C(\bx_n)}\binom{n_\D(\bx_n)}{s_{\D}(\bx_n)}}{\binom{n}{s(\bx_n)}}\indent \forall \bX_n\in\mathcal{X}_n.$$
 This test is exact, i.e., the type~I~error is bounded by (and as close as possible up to a discreteness error to) the significance level~$\alpha$ when~$\pi(\bx)=1/2$ for all~$\bx\in\mathcal{X}$~\citep{agresti1992}.

According to the above testing situation, we can define four OCs, which depend on the parameters~$\btheta$, where we let~$\mathbb{E}^\pi_\btheta$ denote the expectation w.r.t.~$\mathbb{P}^\pi_\btheta$.
\begin{itemize}
  \item  {\bf Patient Benefit}:\\ 
The patient benefit is calculated as~$$ \mathbb{E}^\pi_\btheta[n_{\C}(\bX_n)/n]\cdot \mathbb{I}(\theta_\C >\theta_\D) + \mathbb{E}^\pi_\btheta[1-n_\C(\bX_n)/n]\cdot \mathbb{I}(\theta_\D >\theta_\C) + 1/2\cdot \mathbb{I}(\theta_\C=\theta_\D).$$ This OC represents the patient~benefit in the trial given the parameters, which we want to be high.\\ 
 \item {\bf Rejection rate (RR)}:\\ This OC equals~$\mathbb{P}_\btheta^\pi(\Tau(\bX_n)\leq \alpha)$, i.e., the probability  of rejecting~$H_0$ at the end of the trial. If~$\theta_\C=\theta_\D$, this probability equals the type~I error and we want the rejection rate to be less than~$\alpha$. If~$\theta_\C\neq \theta_\D$, this OC equals the power and we want the rejection rate to be high. \\
\item {\bf Bias}:\\
This OC equals~$\mathbb{E}^\pi_\btheta[\hat{\theta}_\D(\bX_n)-\hat{\theta}_\C(\bX_n)] -  \theta_\D-\theta_\C,$ where for all~$\bx_n\in\mathcal{X}_n$,~$a\in\{\C,\D\}$
$$\hat{\theta}_a(\bx_n) = \begin{cases}
    s_{a}(\bx_n)/n_{a}(\bx_n),\indent&\text{if~$\min_{a}n_{a}(\bx_n)>0$,}\\
     (s_{a}(\bx_n)+1)/(n_{a}(\bx_n)+2),&\text{else.}
\end{cases}$$
The above adjustment to the maximum likelihood estimator~$s_a(\bx_n)/n_a(\bx_n)$ is made in order to provide an estimate when either of the treatment groups contains zero observations. The OC equals the expected error in the estimated treatment effect at the end of the trial, and we want the bias to be as close to zero as possible.\\
\item  {\bf Mean squared error (MSE)}:\\ 
This OC equals~$ \mathbb{E}^\pi_\btheta[(\hat{\theta}_\D(\bX_n)-\hat{\theta}_\C(\bX_n) - (\theta_\D-\theta_\C))^2]$
and corresponds to the quality of the treatment effect estimate~$\hat{\theta}_\D(\bX_n)-\hat{\theta}_\C(\bX_n)$, we want this OC to be low.
\end{itemize}

\section{Constrained Markov decision processes}\label{sect:CMDP}
In this section, we introduce the class of CMDP procedures.
The class of CMDP procedures is based on constrained optimisation. Trial~participants are sequentially allocated treatment, where the probability of allocating a treatment is determined based on a trade-off between exploration and exploitation, while the resulting distribution over states should also be such that certain constraints are satisfied. It is allowed to formulate the constraint under a different transition structure. This makes it possible to define, e.g., type~I~error or power constraints, where optimisation can be performed under a model that does not assume any prior information on~$\btheta$, while the kernel under the null hypothesis would restrict~$\theta_\C=\theta_\D$, and the kernel under the alternative hypothesis could, e.g., assume~$|\theta_\D-\theta_\C|\geq \xi$ for some~$\xi\in(0,1).$

\subsection{Formulation of the optimisation problem}
 We determine a CMDP~procedure according to a Bayesian optimisation problem, where the expected number of successes is maximized under a prior predictive distribution, given a set of policy constraints. 
 We assume an independent~$\text{Beta}(\tilde{s}_{a,0},\tilde{f}_{a,0})$ prior~$\Pi$  for each arm~$a\in\{\C,\D\}$, where~$\text{Beta}(\cdot,\cdot)$ denotes the Beta distribution, and~$\tilde{s}_{a,0},\tilde{f}_{a,0}\in (0,\infty)$ are a prior number of successes and failures.
 The~$\text{Beta}(\tilde{s}_{a,0},\tilde{f}_{a,0})$ prior is a conjugate prior for the Bernoulli distribution, which brings the advantage that the posterior distribution for the success probability is known in closed form.
 As in~\citet{cheng2007optimal}, for a CMDP procedure, the actions~$\delta_t\in[1-p,p]$ correspond to the probability of allocating the next participant to the control arm, where~\mbox{$p\in[1/2,1]$} is the degree of randomization.
Let~\mbox{$\tilde{s}_a(\bx_t) = s_a(\bx_t)+\tilde{s}_{a,0}$} and~$\tilde{n}_a(\bx_t) = n_a(\bx_t)+\tilde{s}_{a,0}+\tilde{f}_{a,0}$   for each arm~\mbox{$a\in\{\C,\D\}$.}
The transition dynamics~\eqref{Transdyd_state} can be made independent of~$\btheta$ by taking the expectation of the transition probabilities with respect to the prior, in which case, following~\citep[Section~2.3]{williamson2017bayesian} and \mbox{ letting~$q(\bx_{t+1},\bx_t,\delta_t)=\mathbb{P}(\bX_{t+1}=\bx_{t+1}\mid \bX_t=\bx_t,\delta_t)$:}
\begin{equation}
q(\bx_{t+1},\bx_t,\delta_t)=\begin{cases}
\delta_t\cdot \tilde{s}_{\C}(\bx_t)/\tilde{n}_{\C}(\bx_t), \indent &\text{ if~$\bx_{t+1} = {\bx}_t +  \partial \bs_\C$,}\\
\delta_t\cdot (1-\tilde{s}_{\C}(\bx_t)/\tilde{n}_{\C}(\bx_t)),\indent &\text{ if~${\bx}_{t+1} = {\bx}_t +  \partial \bff_\C$,}\\
(1-\delta_t)\cdot \tilde{s}_{\D}(\bx_t)/\tilde{n}_{\D}(\bx_t),\indent &\text{ if~${\bx}_{t+1} = {\bx}_t +  \partial \bs_{\D},$}\\
(1-\delta_t)\cdot (1- \tilde{s}_{\D}(\bx_t)/\tilde{n}_{\D}(\bx_t) ),\indent &\text{ if~${\bx}_{t+1} = {\bx}_t +  \partial \bff_{\D}.$}\\
\end{cases}\label{transdyd_state_CMDP}\end{equation}

For all~$\bx_n\in\mathcal{X}_n$,~$t< n$,~$\delta_t,\delta_n\in[1-p,p]$,  and~$\bx_t\in\mathcal{X}_t$ let 
$$r({\bx}_t, \delta_t) = \delta_t\tilde{s}_{\C}(\bx_t)/\tilde{n}_{\C}(\bx_t) + (1-\delta_t)\tilde{s}_{\D}(\bx_t)/\tilde{n}_{\D}(\bx_t),\indent r({\bx}_n, \delta_n) = 0,$$ be the posterior mean rewards after choosing allocation probability~$\delta_t$ for the control treatment in state~${\bx}_t$, for~$t\leq n$. 
A CMDP~procedure now maximizes the total expected reward, i.e., the expected sum of~$r(\bX_t, \delta_t)$ over decision epochs. 

Constraints are enforced ensuring good OCs for the resulting procedure.
The constraints are defined in terms of total expected reward under an alternative prior for~$\btheta$, i.e., defining priors~$\Pi_c$ for~$\btheta$ for~$c$ in a finite countable index set~$\mathcal{C}$,  transition probabilities~$\mathbb{P}_c(\bX_{t+1} = \bx_{t+1}\mid\bX_t=\bx_t,\,\delta_t)$ can be defined according to~\eqref{Transdyd_state} by integrating out~$\btheta$ w.r.t.~$\Pi_c$. Defining rewards~$r_c(\bx_t,\delta_t)$, and values~$V_c\in\mathbb{R}$  for all~$\bx_{t+1},\bx_t\in\mathcal{X},\;\delta_t\in[1-p,p]$ for~$c\in\mathcal{C}\subseteq \mathbb{N}_0$, a CMDP~procedure  is defined as the solution of
 \begin{subequations}
    \begin{align}
    &\underset{\pi}{\max}\,\mathbb{E}^{\pi}\left[\sum_{t=0}^{n}r({\bX}_t,\delta_t)\right]\label{CMDP}\\
    &\textnormal{s.t.}\nonumber\\
&\mathbb{E}^{\pi}_c\left[\sum_{t=0}^{n}r_c({\bX}_t,\delta_t)\right]\leq V_c\indent \forall c\in\mathcal{C}\label{constraints}.
    \end{align}
    \label{CMDP_tot}
  \end{subequations}\\\hspace{-1.5mm}
Problem~\eqref{CMDP_tot} will be referred to as the CMDP problem. As the Constraints~\eqref{constraints} are specified under expectation operators different from the expectation operator in the Objective~\eqref{CMDP}, Problem~\eqref{CMDP_tot} is a generalization of a finite horizon constrained Markov decision problem as defined in~\mbox{\citet{altman1999constrained}.}
We now give three examples of CMDP procedures:

\phantom{.}
\begin{example}[Control of power and type~I~error] \label{example:CMDP-T}
    A downside of RA~procedures is that reliable statistical inference is complicated due to an imbalance and variance over treatment group sizes, as well as dependence in the outcomes induced by the RA procedure, which can lead to a reduction of power and/or type~I~error inflation in comparison to a NA~design~\citep{robertson2023response}.  To test a certain hypothesis on the difference~\mbox{$\theta_\D-\theta_\C$} after collecting data using the CMDP~procedure, the constraints~\eqref{constraints} can be specified such that the Bayesian expected power and type~I~error remain under control. This idea has similarities with Bayesian sample size determination, where Bayesian expected power and type~I~error are used to determine the sample size of a clinical trial based on historical data~\citep{BROWN198729}.

We propose to control the type~I~error by defining a prior~$\Pi_0$ with \mbox{support~$\Theta_0= \{\btheta\in[0,1]^2: \theta_\C = \theta_\D\}$,} agreeing with the null hypothesis~$H_0$. In this evaluation, we let~$\Pi_0$ correspond to a~$\text{Beta}(\tilde{s}_{0}, \tilde{f}_0)$~prior on~$\theta_\C$ with the additional restriction~$\theta_\C=\theta_\D,$ where~$\tilde{s}_0,\tilde{f}_0\in(0,\infty)$. As in~\eqref{transdyd_state_CMDP}, we integrate out the parameter~$\btheta$ in~\eqref{Transdyd_state} but now using the prior~$\Pi_0$. Letting~$\tilde{s}(\bx_t) = s(\bx_t) + \tilde{s}_{0}$,~$n(\bx_t)=n_\C(\bx_t)+n_\D(\bx_t)$,~$\tilde{n}(\bx_t)=n(\bx_t)+\tilde{s}_0+\tilde{f}_0$, \mbox{and~$q_0(\bx_{t+1},\bx_t,\delta_t)=\mathbb{P}_0(\bX_{t+1}=\bx_{t+1}\mid \bX_t=\bx_t,\delta_t)$} we obtain 
for all~${\bx}_t\in{\mathcal{X}_t},{\bx}_{t+1}\in{\mathcal{X}_{t+1}}$ 
~$$q_0({\bx}_{t+1}{\bx}_t,\delta_t) = \begin{cases}
\delta_t\cdot \tilde{s}(\bx_t)/\tilde{n}(\bx_t), \indent &\text{ if~$\bx_{t+1} = \bx_t +  \partial \bs_\C$,}\\
\delta_t\cdot (1-\tilde{s}(\bx_t)/\tilde{n}(\bx_t)),\indent &\text{ if~$\bx_{t+1} = \bx_t +  \partial \bff_\C$,}\\
(1-\delta_t)\cdot \tilde{s}(\bx_t)/\tilde{n}(\bx_t),\indent &\text{ if~$\bx_{t+1} = \bx_t +  \partial \bs_{\D}$,}\\
(1-\delta_t)\cdot (1- \tilde{s}(\bx_t)/\tilde{n}(\bx_t) ),\indent &\text{ if~$\bx_{t+1} = \bx_t +  \partial \bff_{\D}.$}\\
\end{cases}$$
We enforce the type~I~error constraint
\begin{equation}\mathbb{P}^\pi_0(\Tau(\bX_n)\leq \alpha)=\int_{[0,1]^2} \mathbb{P}^\pi_{\btheta}(\Tau(\bX_n)\leq \alpha)\Pi_0(d\btheta)\leq \alpha^*.\label{typeIerror}\end{equation}
We see from~\eqref{typeIerror}  that the constraint enforces that the weighted average of the frequentist type~I~error~$\mathbb{P}^\pi_{\btheta}(\Tau(\bX_n)\leq \alpha)$ is bounded by~$\alpha^*$, hence to \mbox{ensure~$\mathbb{P}^\pi_0(\Tau(\bX_n)\leq \alpha\mid \btheta)\leq \alpha$} for all~$\btheta$, it is necessary that~$\alpha^*\leq \alpha$ above.

Similarly, we can define a prior~$\Pi_1$ with support in~$[0,1]^2$, and, given a required maximum type~II~error~$\beta\in(0,1)$, the power constraint becomes
$$\mathbb{P}^\pi_1(\Tau(\bX_n)\leq \alpha)\geq 1-\beta.$$
The resulting optimisation problem, the solution of which is denoted  as CMDP-T procedure, can be written as 
\begin{subequations}
    \begin{align}
    &\underset{\pi}{\max}\,\mathbb{E}^{\pi}\left[\sum_{t=0}^{n}r(\bX_t,\delta_t)\right]\\
    &\textnormal{s.t.}\nonumber\\
    &\mathbb{P}^\pi_0(\Tau(\bX_n)\leq \alpha)\leq \alpha^*,\\
&\mathbb{P}_1^\pi(\Tau(\bX_n)\leq \alpha)\geq 1-\beta.
\end{align}
\label{CMDP_TEST}
\end{subequations}\\\hspace{-1.5mm}
Note that~\eqref{CMDP_TEST} corresponds to~\eqref{CMDP_tot} with~$r_c(\bx_t,\delta_t) = 0$ for~$\bx_t\in\mathcal{X}_t$ and~$t<n$, 
~$V_0 = \alpha^*$,~$V_1=-(1-\beta)$ \mbox{and~$r_c(\bx_n,\delta_n) = (-1)^{\mathbb{I}(c=1)}\mathbb{I}(\Tau(\bX_n)\leq \alpha)$} for~$\bx_n\in\mathcal{X}_n$.
\end{example}
\phantom{.}
\begin{example}[Control of estimation error] \label{example: CMDP-E}
    Next to difficulties in testing, using an RA procedure can also lead to large errors when estimating treatment effects~\citep{Bowden2017unbiased}.
To improve the estimate of the treatment \mbox{effect~$\theta_\D-\theta_\C$} after collecting data using the CMDP~procedure, we specify the constraints~\eqref{constraints} such that the MSE remains under control. 

 We optimise the successes incurred under the~$\text{Beta}(\tilde{s}_{a,0},\tilde{f}_{a,0})$ prior introduced in Section~\ref{sect:model}, under the constraint that the posterior MSE is small.  
 In order to control the shape of the MSE curve over~$[0,1]^2$,
 we discretize~$[0,1]^2$ into a  collection of disjoint two-dimensional intervals~$$\mathcal{S}=\{[\theta_\C^\ell,\theta_\C^u)\times[\theta_\D^\ell,\theta_\D^u):0\leq\theta_a^\ell\leq \theta_a^u\leq 1\;\forall a\in\{\C,\D\}\}$$ 
 such that~$\cup_{\sigma\in\mathcal{S}}\sigma=[0,1)^2$ 
 and define~$\Pi_\sigma$ to be a~$\text{Beta}(\tilde{s}^E_{a,0},\tilde{f}^E_{a,0})$
 prior truncated to~$\sigma$  for all~$\sigma\in\mathcal{S}$,~$a\in\{\C,\D\}.$ 
 By the law of total expectation, letting~$\Pi_\sigma(\btheta\mid \bX_n)$ denote the posterior distribution under prior~$\Pi_\sigma$, the posterior MSE for a policy~$\pi$ over~$\sigma$ in~$\mathcal{S}$ can be expressed~as 
 \begin{equation}
\mathbb{E}_\sigma^\pi[\textstyle\int_{\sigma}(\hat{\theta}_\D(\bX_n) - \hat{\theta}_\C(\bX_n) - (\theta_\D - \theta_\C))^2d\Pi_\sigma(\btheta\mid \bX_n)].\label{eqn:MSE}
 \end{equation}
 The expression in \eqref{eqn:MSE} is the expectation of a function of the final state~$\bX_n$, which does not depend on the policy as all allocations have been realised, hence a bound on this quantity can be written as~\eqref{constraints}. To see whether policies with high patient~benefit, type~I~error control and high power, as well as low MSE can be found, we furthermore add the power and type~I~error constraints as in the CMDP-T formulation~\eqref{CMDP_TEST}.
The resulting optimisation problem, the solution of which is denoted  as CMDP-E RA procedure, can be written as 

\begin{subequations}
    \begin{align}
    &\underset{\pi}{\max}\,\mathbb{E}^{\pi}\left[\sum_{t=0}^{n}r(\bX_t,\delta_t)\right]\\
    &\textnormal{s.t.}\nonumber\\
    &\mathbb{P}^\pi_0(\Tau(\bX_n)\leq \alpha)\leq \alpha^*,\\
    &\mathbb{P}_1^\pi(\Tau(\bX_n)\leq \alpha)\geq 1-\beta,\\
&\mathbb{E}_\sigma^\pi[\textstyle\int_{\sigma}(\hat{\theta}_\D(\bX_n) - \hat{\theta}_\C(\bX_n) - (\theta_\D - \theta_\C))^2d\Pi_\sigma(\btheta\mid \bX_n)]\leq V_\sigma,\indent \forall \sigma\in\mathcal{S}.\label{constraints_squares}
\end{align}
\label{CMDP_EST}
\end{subequations}\\\hspace{-1.5mm}
The priors under the constraints~\eqref{constraints_squares} are designed in such a way that the constraints reflect the average behaviour of the policy on a specific part of the parameter space~$[0,1]^2$. 
\end{example}
\phantom{.}
\begin{example}[Robustness to prior misspecification] \label{example:CMDP-R}
    In clinical trials, historical data and expert opinion are often available on one or both treatments. In Bayesian RA procedures, historical data can be used to construct a prior distribution on the success probabilities, which may be leveraged to increase patient~benefit. However, if the actual treatment effects have low probability mass under the prior, i.e., when there is prior misspecification, the number of~participants allocated to the optimal arm can be lower than under a less informative prior. The constraints~\eqref{constraints} can be specified such that the RA procedure is robust against prior misspecification, by controlling the number of treatment failures~$1-Y_{A_t,t}$ under a less informative prior distribution. 

We assume that the incorporation of historical data and expert opinion leads to an independent~$\text{Beta}(\tilde{s}_{a,0},\tilde{f}_{a,0})$ prior on the success probability for each arm~$a$ in the trial, corresponding to the prior~$\Pi$ used for optimisation in~\eqref{CMDP_tot}.
Let~$\Pi_\text{LI}$ be a prior on the success probabilities that is less informative, e.g., a uniform prior, let ~$$r_\text{LI}(\bx_t,\delta_t) = -\delta_t \mathbb{E}_\text{LI}[\theta_\C\mid \bX_t=\bx_t] + (1-\delta_t) \mathbb{E}_\text{LI}[\theta_\D\mid \bX_t=\bx_t],\;\;\;\;\forall t<n,$$
and~$r_\text{LI}(\bx_n,\delta_n )=0$ for all~$t\in\{0,\dots, n\}$,~$\bx_t\in\mathcal{X}_t$,~$\bx_n\in\mathcal{X}_n$,~$\delta_t,\delta_n\in[1-p,p]$.
We optimise the successes incurred under~$\Pi$ under the constraint that the total expected number of successes under~$\Pi_\text{LI}$ is within a percentage of the maximum under that prior, i.e., letting~\mbox{$v_\text{LI}= \max_\pi \mathbb{E}_\text{LI}^{\pi}\left[\sum_{t=0}^{n}r_\text{LI}(\bX_t,\delta_t)\right]$} and~$\xi\in(0,1)$, the resulting optimisation problem, the solution of which is denoted  as CMDP-R RA procedure, can be written as 
\begin{subequations}
    \begin{align}
    &\underset{\pi}{\max}\,\mathbb{E}^{\pi}\left[\sum_{t=0}^{n}r(\bX_t,\delta_t)\right]\\
    &\textnormal{s.t.}\nonumber\\
    & -\mathbb{E}_\text{LI}^{\pi}\left[\sum_{t=0}^{n}r_\text{LI}(\bX_t,\delta_t)\right] \leq -\xi\cdot v_\text{LI}  ,
\end{align}
\label{CMDP_ROBUST}
\end{subequations}\\\hspace{-1.5mm}
which is of the form~\eqref{CMDP_tot} with~$V_\text{LI} = -\xi\cdot v_\text{LI}$. Note that the unconstrained version of~\eqref{CMDP_ROBUST} would correspond to the RA procedure in~\citet{cheng2007optimal} when not considering~participants outside the trial. 
\end{example}

\subsection{Properties of the constrained Markov decision process problem}
This section provides theoretical properties of the CMDP~problem~\eqref{CMDP_tot} and introduces a computational approach to obtain an optimal policy~$\pi^*$ for~\eqref{CMDP_tot} which involves solving a linear program~(LP). Section~\ref{sect:solutionmethod} provides a computationally more efficient procedure using backward recursion, which finds a feasible, but possibly suboptimal solution for~\eqref{CMDP_tot}, for which the (relative) optimality gap of the approximation can be calculated.

We first show that problem~\eqref{CMDP_tot} can be rewritten to a standard finite-horizon constrained Markov decision process. First, we need the following lemma.

\phantom{.}
\begin{lemma}~\label{lemma:datalikelihood_Bernoulli_RA}
  For all~$c\in\mathcal{C}$,~$t\leq n$, and  states~$\bx_t \in\mathcal{X}_t$
\begin{equation}\mathbb{P}^\pi(\bX_t=\bx_t)= g_t^\pi(\bx_t)q(\bx_t),\;\;\;\;\mathbb{P}_c^\pi(\bX_t=\bx_t)= g_t^\pi(\bx_t)q_c(\bx_t),\label{expression_likelihood}\end{equation}
where
\begin{align}
 q(\bx_t) &=  \int_{[0,1]^2}\prod_{a\in\{\C,\D\}}\theta_a^{s_a(\bx_t)}(1-\theta_a)^{n_{a}(\bx_t)-s_{a}(\bx_t)}\Pi(d\btheta),\label{eqn:q_firstmeas}\\
    q_c(\bx_t) &= \int_{[0,1]^2}\prod_{a\in\{\C,\D\}}\theta_a^{s_a(\bx_t)}(1-\theta_a)^{n_{a}(\bx_t)-s_{a}(\bx_t)}\Pi_c(d\btheta),\nonumber
\end{align}
and where~$g_t^\pi$ is defined recursively by
 \begin{align*}
g_0^\pi(\bx_0)&=1,\\
     \indent g_{t}^\pi(\bx_t) &= \sum_{\substack{a\in\{\C,\D\}\\\partial\bx_a\in\{\partial\bs_a,\partial{\bff}_a\}}}g_{t-1}^\pi(\bx_{t}-\partial\bx_a)\pi(\bx_{t}-\partial\bx_a)^{\mathbb{I}(a=\C)}(1-\pi(\bx_{t}-\partial\bx_a))^{\mathbb{I}(a=\D)},\label{eqn:defGcoefs}
\end{align*} 
for all~$\bx_t\in\mathcal{X}_t$, ~$t\in\mathbb{N}$.

\end{lemma}

\begin{proof}
From \cite[Equation (1)]{yi2013exact}, we have
\begin{equation}\mathbb{P}_\btheta^\pi(\bX_t=\bx_t)= g_t^\pi(\bx_t)\prod_{a\in\{\C,\D\}}\theta_a^{s_{a,t}}(1-\theta_a)^{n_{a,t}-s_{a,t}}\indent\forall t.\label{deflikelihood}\end{equation}
The statement of the lemma follows as
\begin{align*}
    \mathbb{P}^\pi(\bX_t=\bX_t)  = \int_{[0,1]^2}\mathbb{P}_\btheta^\pi(\bX_t=\bx_t)\Pi(d\btheta),\\\mathbb{P}_c^\pi(\bX_t=\bX_t)=\int_{[0,1]^2}\mathbb{P}_\btheta^\pi(\bX_t=\bx_t)\Pi_c(d\btheta).
\end{align*}
\end{proof}
In Theorem~\ref{thm:CMDP_reformulated} below, Lemma~\ref{lemma:datalikelihood_Bernoulli_RA} is used to rewrite the constraints~\eqref{constraints} under the same expectation operator as the Objective~\eqref{CMDP} using a change of measure.

\phantom{.}
\begin{theorem}\label{thm:CMDP_reformulated}
    If~$q({\bx})=0$ implies~$q_c(\bx)=0$ for all~$c$ and~$\bx\in\mathcal{X}$, problem~\eqref{CMDP_tot} can be rewritten as
    \begin{subequations}
    \begin{align}
    &\underset{\pi}{\max}\,\mathbb{E}^{\pi}\left[\sum_{t=0}^{n}r({\bX}_t,\delta_t)\right]\\
    &\textnormal{s.t.}\nonumber\\
&\mathbb{E}^{\pi}\left[\sum_{t=0}^{n}\tilde{r}_c({\bX}_t,\delta_t)\right]\leq V_c\indent \forall c\in\mathcal{C},\label{CMDP2_cons}
\end{align}
\label{CMDP2_tot}
 \end{subequations}\\\hspace{-1.5mm}
where 
$$\tilde{r}_c({\bx}_t,\delta_t) =\begin{cases} r_c(\bx_t,\delta_t)q_c(\bx_t)/q(\bx_t),\indent &\text{if~$q(\bx_t)>0$},\\
0,&\text{else.}\end{cases}$$
\end{theorem}
\begin{proof}
    The result follows as, by Lemma \ref{lemma:datalikelihood_Bernoulli_RA}, 
    \begin{align*}
&\mathbb{E}^{\pi}_c\left[\sum_{t=0}^{n}r_c(\bX_t,\delta_t)\right]=\sum_{t=0}^{n}\sum_{{\bx}_t\in\mathcal{X}_t}\mathbb{P}^\pi_c({\bX}_t={\bx}_t)r_c( {\bx}_t, \pi(\bx_t)) \\&\stackrel{\eqref{expression_likelihood}}{=}\sum_{t=0}^{n}\sum_{\substack{{\bx}_t\in\mathcal{X}_t\\ q({\bx}_t)>0}}g_t^\pi({\bx}_t)q_c({\bx}_t)r_c( {\bx}_t, \pi({\bx}_t))=\sum_{t=0}^{n}\sum_{\substack{{\bx}_t\in\mathcal{X}_t\\ q({\bx}_t)>0}}g_t^\pi({\bx}_t)q({\bx}_t)\tilde{r}_c( \bx_t, \pi(\bx_t)) \\&\stackrel{\eqref{expression_likelihood}}{=} \sum_{t=0}^{n}\sum_{\substack{{\bx}_t\in{\mathcal{X}_t}}}\mathbb{P}^\pi({\bX}_t={\bx}_t)\tilde{r}_c( \bx_t, \pi(\bx_t)) = \mathbb{E}^{\pi}\left[\sum_{t=0}^n\tilde{r}_c(\bX_t,\delta_t)\right] .
    \end{align*} 
\end{proof}

The next result follows from~\citet[Theorem 2.1 and Theorem 3.8]{altman1999constrained}, by reformulating the finite horizon constrained~MDP~\eqref{CMDP_tot} to a discounted infinite horizon constrained~MDP.

\phantom{.}
\begin{theorem}\label{thm:altman}
    If~$q({\bx})=0$ implies~$q_c(\bx)=0$, then for any feasible policy 
   which depends on the full history of states and actions, there is a feasible randomized Markov policy~$\pi':{\mathcal{X}}\mapsto [1-p,p]$ inducing at least the same total expected reward in~\eqref{CMDP_tot}, i.e., under Definition~2.2 in \citet{altman1999constrained}, the class of randomized Markov policies dominates the class of history-dependent policies for the CMDP. Furthermore, there exists an optimal Markov policy such that the number of random actions given the current state is at most~$|\mathcal{C}|$.
\end{theorem}
\phantom{.}

 Observe that randomization, a property that is usually desired in clinical trials, arises naturally for CMDP procedures according to Theorem~\ref{thm:altman}, i.e., without restricting the action space.
 Note that this is not the case for the MDP~BRA~procedures introduced in~\citet{cheng2007optimal} or~\citet{williamson2017bayesian},  where \mbox{setting~$p=1.00$} leads to a DRA~\citep[Theorem 1]{cheng2007optimal}, while  randomized optimal actions are enforced when \mbox{setting~$p<1$.} 
Furthermore, for the model introduced in Section~\ref{sect:model}, the condition that ~$q({\bx})=0$ implies~$q_c(\bx)=0$ is only violated when~$\Pi$ concentrates on~$\{0,1\}^2$ as, following~\eqref{eqn:q_firstmeas},~$q(\bx)$ is only zero if either \mbox{$\theta_\C^{s_\C(\bx)}$}, \mbox{$(1-\theta_\C)^{n_\C(\bx)-s_\C(\bx)},~$} \mbox{$\theta_\D^{s_\D(\bx)},$} \mbox{or $(1-\theta_\D)^{n_\D(\bx)-s_\D(\bx)}$} is zero with probability one under the prior~$\Pi$. When defining a CMDP procedure, such measures will usually not be considered for~$\Pi$. 

We now turn to solving~\eqref{CMDP2_tot}, and hence~\eqref{CMDP_tot}, by reformulating~\eqref{CMDP2_tot} as an LP. Let~$\mathcal{X}_{<n}=\cup_{t=0}^{n-1}\mathcal{X}_t$, \mbox{$d_{<n}=|\mathcal{X}_{<n}|$,}~$d_n=|\mathcal{X}_n|$,~$d_{\leq n} = |\mathcal{X}|$, and
$d=d_n + 2d_{<n}$ be the number of state-action pairs in the CMDP. Let 
$$\bA = \begin{bmatrix}
    \bI_{d_{<n}} &\bI_{d_{<n}}& \mathcal{O}_{d_{<n},d_n}\\
    \mathcal{O}_{d_n,d_{<n}} &\mathcal{O}_{d_n,d_{<n}}&\bI_{d_n}
\end{bmatrix} -\begin{bmatrix}
    \bP_C &\bP_D& \mathcal{O}_{d_{\leq n},d_n}
\end{bmatrix} \in\mathbb{R}^{d_{\leq n}\times d},$$
where~$\bI_k,\mathcal{O}_{k,\ell}$ are the identity and zero matrix in~$\mathbb{R}^{k\times k},\mathbb{R}^{k\times \ell}$  for all~$k,\ell\in\mathbb{N}$, and~$\bP_C,\bP_D\in\mathbb{R}^{d_{\leq n}\times d_{<n}}$ are matrices such that 
\begin{align*}[\bP_C]_{i(\bx),i(\bx')} &= \mathbb{P}(\bX_{t+1} = \bx\mid\bX_t=\bx',\,\delta_t=p),\indent &&\forall \bx\in\mathcal{X},\bx'\in\mathcal{X}_{<n},\\
[\bP_D]_{i(\bx),i(\bx')} &= \mathbb{P}(\bX_{t+1} = \bx\mid\bX_t=\bx',\,\delta_t=1-p),\indent &&\forall \bx\in\mathcal{X},\bx'\in\mathcal{X}_{<n},
\end{align*}
where~$i:\mathcal{X}\mapsto \{1,\dots, d_{\leq n}\}$ is a storage mapping function \citep{jacko2019binarybandit} for the states such \mbox{that~$i(\mathcal{X}_{<n})= \{1,\dots, d_{<n}\}$.}
Let~$\br\in\mathbb{R}^{d}$ be a vector defined as~$$r_j = \sum_{\bx\in\mathcal{X}}r(\bx,p)\mathbb{I}(j=i(\bx))+r(\bx, 1-p)\mathbb{I}(j=i(\bx)+d_{<n})
.$$
Let~$\tilde{\br}_c$ be a similarly defined vector for all~$c\in\mathcal{C}$, and 
$\bb$ be a vector denoting the initial distribution, defined such that~$b_j = \mathbb{I}(j=i(\bx_0))$ for all~$j.$
The next theorem states that the optimal policy for~\eqref{CMDP2_tot}, hence~\eqref{CMDP_tot}, can be found by solving a linear~program. 

\phantom{.}
\begin{theorem}
If~\eqref{CMDP_tot} is feasible and ~$q({\bx})=0$ implies~$q_c(\bx)=0$, an optimal policy~$\pi^*$ for~\eqref{CMDP_tot} is given by~$$\pi^*(\bx)=\begin{cases}
    \frac{p\mu_{i(\bx)}+(1-p)\mu_{i(\bx) + d_{<n}}}{(\mu_{i(\bx)} + \mu_{i(\bx) + d_{<n}})},\indent&\text{if~$\mu_{i(\bx)} + \mu_{i(\bx) + d_{<n}}>0$,}\\ 1/2 ,&\text{else,}
\end{cases}$$
where the vector~$\bmu$ is the solution to
\begin{subequations}
     \begin{align}
&\max_{\bmu\in\mathbb{R}^{d}}\bmu^\top \br \label{obj_LP}\\
&\textnormal{s.t.}\nonumber\\
& \bA\bmu = \bb,\label{lp:axleqb}\\
  &\bmu^\top \tilde{\br}_c \leq V_c \indent \forall c\in\mathcal{C}\label{lp:additionalcondstr},\\
  &\bmu\geq \bm 0\label{lp:nonneg}.
\end{align}
\label{LPform}
\end{subequations}\\\hspace{-1.5mm}
\end{theorem}
\begin{proof}
    The result can be shown by writing the dual linear program formulation of the Markov decision process~\eqref{CMDP2_tot} without constraints~\citep{puterman2014markov}, consisting of~\eqref{obj_LP},~\eqref{lp:axleqb} and~\eqref{lp:nonneg}, where~$\delta_t\in[1-p,p]$ can be viewed as a Randomization over actions~$a_t\in\{1-p,p\}$.
The statement of the theorem follows by adding the constraints in~\eqref{CMDP2_cons} as~\eqref{lp:additionalcondstr}, which follow as~$\bmu$ is the probability vector for the state-action pairs in the CMDP, a fact that is also used in determining the optimal policy~$\pi^*$.
\end{proof}

\subsection{Proposed solution method} \label{sect:solutionmethod}
Solving~\eqref{LPform} can be computationally heavy for large values of~$d$.  The next theorem provides a way to find a solution to~\eqref{CMDP_tot} using backward recursion, which is computationally more  tractable. Conditions are given under which the obtained solution is optimal.

\phantom{.}
\begin{theorem}\label{Thm:minmax}
    Let~$V$ be the value of~\eqref{CMDP_tot}. If~\eqref{CMDP_tot} is feasible and~$q({\bx})=0$ \mbox{implies~$q_c(\bx)=0$,} we have 
    \begin{equation}
V = \min_{\blambda\in\mathbb{R}^C_+} L(\blambda), \;\; L(\blambda)=\max_{\pi}\mathbb{E}^\pi\left[\sum_{t=0}^{n}{r}(\bX_t, \delta_t) + \sum_c \lambda_c \left(V_c/n - \tilde{r}_c(\bX_t, \delta_t)  \right)\right].\label{Eqn: minmaxV}
    \end{equation}
    If~\eqref{Eqn: minmaxV} yields an optimal point~$(\hat{\blambda},\hat{\pi})$ such that \begin{equation} \sum_c \hat{\lambda}_c \mathbb{E}^{\hat{\pi}}\left[\sum_{t=0}^{n}\left(V_c/n - \tilde{r}_c(\bX_t, \delta_t)  \right)\right]=0,\label{KKTcond}\end{equation}
   i.e.,~$(\hat{\blambda}, \hat{\pi})$ satisfies the Karush-Kuhn Tucker~(KKT)~conditions~\citep{kuhn2014nonlinear},  then~$\hat{\pi}$ is an optimal solution to~\eqref{CMDP_tot}. 
\end{theorem}
\begin{proof}
 Using Lagrange multipliers, letting~$\tilde{\bR},\;\tilde{\bV}$ be the horizontal concatenation of the vectors~$\tilde{\br}_c$ and constants~$V_c$, we reformulate~\eqref{LPform} to
 \begin{subequations}
     \begin{align}
&\max_{\bmu\in\mathbb{R}^{d}}\min_{\blambda\geq \bm 0}\,\bmu^\top \br + \blambda^\top(\tilde{\bV}-\tilde{\bR}\bmu    ) \\
&\textnormal{s.t.}\nonumber\\
& \bA\bmu = \bb,\\
  &\bmu\geq \bm 0.
\end{align}
\label{LPform2}
 \end{subequations}\\\hspace{-1.5mm}
By \citep[Lemma  9.2]{altman1999constrained}, we have that~\eqref{LPform2} is equivalent to 
\begin{subequations}
    \begin{align}
&\min_{\blambda\geq \bm 0}\max_{\bmu\in\mathbb{R}^{d}}\,\bmu^\top \br + \blambda^\top(\tilde{\bV}-\tilde{\bR}\bmu   ) \\
&\textnormal{s.t.}\nonumber\\
& \bA\bmu = \bb,\\
  &\bmu\geq \bm 0.
\end{align}
\label{LPform2_2}
\end{subequations}\\\hspace{-1.5mm}
Choosing~$\blambda$ fixed in~\eqref{LPform2_2} results in the LP that returns~$L(\blambda)$,
from which~\eqref{Eqn: minmaxV} follows. The conditions for an optimal solution follow from~\citet{kuhn2014nonlinear}.
\end{proof}

The function~$L$ is a convex function in~$\blambda$, as it is the maximum over a set of affine functions in~$\blambda$ (one for each~$\pi$).
Backward recursion \citep{ puterman2014markov} can be performed to find the value~$L(\blambda)$ and a maximizer~$\pi^*_\blambda$ of~\eqref{Eqn: minmaxV} (letting~$\pi^*_\blambda(\bx_t)=1/2$ in case of ties). Hence, without loss of generality, the set of randomized Markov  policies optimized over in~\eqref{Eqn: minmaxV} can be restricted to the set~$\mathcal{P}$ of \mbox{policies~$\pi:\mathcal{X}\mapsto\{1-p,1/2,p\}$} when determining~$V$.
The optimal policy~$\pi^*_\blambda$ in this set can  be used to determine a subgradient~$\nabla L(\blambda)~$ of~$L$ such \mbox{that~$\nabla L(\blambda)_c= V_c-\sum_{t=0}^n\tilde{r}_c(\bX_t,\pi_\blambda^*(\bX_t))$} for all~$c\in\mathcal{C}$.
We propose Algorithm~\ref{alg:Lshaped} to find~$V$ up to a given numerical precision, which corresponds to a cutting plane method~\citep{kelley1960cutting}. In case of no feasible solution, the algorithm returns~$f^*=-\infty.$ Otherwise, as~$L$ is the maximum over a finite set of affine functions in~$\blambda$, one for each~$\pi\in\mathcal{P}$, the algorithm finds the minimizer after a finite number of iterations. However, due to the (possibly) large amount of affine functions included in the objective, this could take a long time, while a small error~$\epsilon_\text{tol}$ may usually be sufficient. If the policy found under backward recursion is infeasible, the elements of~$\blambda$ corresponding to the infeasible constraints are multiplied by a factor~$(1+\varphi)$ until the policy becomes feasible.
Letting~$\hat{\pi}$ be the resulting policy, the relative optimality gap can be found \mbox{as~$(V - \mathbb{E}^{\hat{\pi}}[\sum_{t=0}^nr(\bX_t,\delta_t)])/V$.}  If~\eqref{KKTcond} is satisfied, it can be verified that~$\hat{\pi}$ is the unique optimiser of~\eqref{CMDP_tot} by verifying \mbox{that~$\hat{\pi}(\bx)\in\{1-p,p\}$} for all~$\bx\in\mathcal{X}$.

\FloatBarrier
\begin{algorithm}[htb]
		\caption{Cutting plane algorithm for calculating V}\label{alg:Lshaped}
		\begin{algorithmic}[1]
			\Inputs{$L,\nabla L, |\mathcal{C}|,\epsilon_\text{tol}$;}
                \Initialize{Set~$\bA = \bm 0^\top_{|\mathcal{C}|+1}$,~$b = 0$,~$\epsilon = \infty,\;f^* = \infty$;}  
                \While{$\epsilon > \epsilon_\text{tol}$ and~$f^*\neq-\infty$}
                \State Solve \vspace*{-3mm}\begin{align*}
                &\hspace*{-50mm}\min_{\bx\in\mathbb{R}^{|\mathcal{C}|+1}} x_1\\
                    &\hspace*{-50mm}\textnormal{s.t.}\\
                    &\hspace*{-50mm}\bA\bx\leq \bb\\
                    &\hspace*{-50mm} \bx \geq \bm 0;
                \end{align*}
                \State Set~$\ubar{f}=x_1,\;\blambda^* = [x_2,\dots, x_{|\mathcal{C}|+1}], \;f^* = \min(f^*, L(\blambda^*)),\; \epsilon = f^*- \ubar{f}, \newline \; \bg~=~\nabla L(\blambda^*);$
                \State Redefine~$\bA = [\bA;[-1, \bg^\top]]$,\;~$\bb = [\bb; -L(\blambda^*) + \bg^\top\blambda^*]$ ($[\bv;\bw]$ denotes  concatenation); 
                \If{$L(\blambda^*)\leq0$}
                    \State~$f^*=-\infty$,~$\blambda^*=\emptyset$;
                    
                \EndIf
                \EndWhile
			\\{\bf Outputs:}~$f^*, \blambda^*$;
		\end{algorithmic}
	\end{algorithm}
\FloatBarrier

\section{Applications}\label{sect:applications}

In this section, we evaluate the performance of CMDP procedures in three applications. First, we evaluate the performance of CMDP-T from Example~\ref{example:CMDP-T}. Second, we evaluate the performance of CMDP-E from Example~\ref{example: CMDP-E}. Third, we evaluate the performance of CMDP-R from Example~\ref{example:CMDP-R}. All CMDP procedures are evaluated according to the OCs introduced in Section~\ref{sect:ocs}, which 
are calculated using forward recursion, instead of simulation, to determine the distribution of~$\bX_n$.

We compare the performance of the CMDP procedures with three other Markov RA procedures known from literature~\citep{williamson2017bayesian}, namely:
\begin{itemize}
    \item {\bf Equal randomization (ER)}: \\$\pi(\bx)=1/2$ for all~$\bx\in\mathcal{X}$.\\
    \item {\bf Dynamic Programming (DP)}:\\ This is the RA procedure~$\pi$ found from the unconstrained version of~\eqref{CMDP_tot} when taking~$p=1.00$.\\
    \item {\bf Constrained Randomized Dynamic Programming (CRDP):}\\
    This is the RA procedure proposed in~\citep{williamson2017bayesian} that follows from the unconstrained version of~\eqref{CMDP_tot} when taking~$p=0.9$ and adding a penalty equal to~$-n$ to the objective whenever~$\min_a n_a(\bx_n)<0.15n$. 
\end{itemize}

The code for calculation of the CMDP policies and the OCs uses an efficient implementation of backward and forward recursion based on~\citet{jacko2019binarybandit},  i.e., uses the conservation law for the states, a storage mapping function, and overwrites elements of the value function in backward recursion and probability distribution in forward recursion that are not used further in the algorithm. 
All experiments are performed in Julia~1.9.0 on a laptop with Intel\textsuperscript{\textregistered} Core\textsuperscript{TM} \hbox{i7-9750H~CPU} with~2.60GHz clock speed and~16 GB RAM. 
The absolute tolerance~$\epsilon_\text{tol}$ for Algorithm~\ref{alg:Lshaped} is set to~$10^{-9}$ for each evaluation. In order to obtain a feasible policy in each considered scenario, we set~$\varphi=0.01$. 
For the evaluation, we let~$\theta_\D\in\{0.00,0.01\dots, 1.00\}$ while fixing~$\theta_\C=0.5.$ In the numerical evaluation, we set the significance level~$\alpha~$ to~$ 0.1$, in agreement with the comparison in~\citet{williamson2017bayesian}. For all CMDP procedures, we set~$p=0.95$, while in the appendices results are also shown for~$p=1.00$. From now on, RA procedures based on a CMDP problem and solved using Algorithm~1 will also be denoted CMDP RA procedures.

\subsection{Application 1: Control of power and type~I~error}

\newcommand{\rewCRDPSMALLHOR}{45.3}
\newcommand{\timeCRDPSMALLHOR}{0.165}
\newcommand{\rewCRDPhighHOR}{122}
\newcommand{\timeCRDPhighHOR}{1.30}

First, we evaluate the CMDP-T procedure (Example~\ref{example:CMDP-T}) for~$n=75$. We set the prior~$\Pi_1$ for the power constraint equal to~$\Pi$.
After evaluating several choices based on the OCs, we set~$\alpha^* = 0.05$,~$\beta=0.4$. The prior parameters~$\tilde{s}_{a,0}, \tilde{f}_{a,0}$ for~$\Pi_0$ are set to~1, as are the parameters for the prior~$\Pi$, hence both~$\Pi_0$ and~$\Pi$ are uniform priors on their respective supports.
Computation of the CMDP-T policy took 7.44 seconds and resulted in a relative optimality gap of~$1.34\cdot 10^{-4}$ and total expected reward 46.9 under the prior~$\Pi$ in comparison to~$\timeCRDPSMALLHOR$ seconds for CRDP with total expected reward~$\rewCRDPSMALLHOR$.
 Figure \ref{results_N75} shows the patient benefit, RR,  bias, and MSE for the policies ER, DP, CRDP and \mbox{CMPD-T} for different values of~$\theta_\D$ and~$\theta_\C=0.5$. The same results for policies ER, DP, and CRDP can also be found in~\citet{williamson2017bayesian}, where the sign of the bias is reversed because we chose to vary the first success rate used in the treatment effect instead of the second one. Figure \ref{results_N75} shows that the CMDP-T~procedure performs very well in terms of patient~benefit, having performance in-between CRDP and DP for all evaluated parameter values. Furthermore, the CMDP-T~procedure has a similar performance in terms of power, while having a worse performance in terms of bias and MSE when compared to CRDP. It is hence seen for this scenario that a trade-off was made between patient~benefit and OCs bias and MSE for CMDP-T, keeping the rejection rate roughly the same. The bias and MSE of CMDP-T are however not as extreme as for DP, and closer to the bias and MSE of CRDP.

Second, we evaluate the results for CMDP-T for~$n=200$.
After evaluating several choices based on the OCs, we set~$\alpha^* = 0.07$, and~$\beta=0.23$. Again, the prior parameters~$\tilde{s}_{a,0}, \tilde{f}_{a,0}$ are set to 1, as are the parameters for~$\Pi$, hence both priors correspond to uniform distributions on their respective supports. Computation of \mbox{the CMDP-T policy} took~296.3~seconds, and resulted in a relative optimality gap of~$5.09\cdot 10^{-5}$ and total expected reward 128 under the prior~$\Pi$ in comparison to~$\timeCRDPhighHOR$~seconds for CRDP with total expected reward~$\rewCRDPhighHOR$.
Figure \ref{results_N200} shows the results. These results were not shown in~\citet{williamson2017bayesian}, where the maximum trial size considered was~$n=100$~participants as the focus was on rare disease trials with a small amount of~participants. The figure shows that in terms of patient~benefit, \mbox{the CRDP procedure} ends up around a value of~$0.85$ for~$\theta_\D$ close to zero and one, which is (roughly) the maximal value this RA procedure can attain, due to the penalty incurred when either~$N_{\C,n}/n$ or~$N_{\D,n}/n$ are lower than~$0.15$.
The~DP \mbox{and CMDP-T procedures} end up at higher values, around~$0.95$ (maximum for \mbox{CMDP-T}) and~$1.00$ (maximum for DP) for~$\theta_\D$ close to zero and one. The power plots show that the power quickly grows to 1.00 in~$|\theta_\D-\theta_\C|$ for policies ER, CRDP and \mbox{CMDP-T} where a power of~$80\%$ is roughly attained when~$|\theta_\D-\theta_\C|\geq 0.25$, while for the DP policy, the power is slightly higher than in Figure~\ref{results_N75}, with highly irregular behaviour around high values of~$\theta_\D.$ 

In conclusion, for~$n=200$ the CMDP-T~procedure   outperforms CRDP in terms of patient~benefit and power, while CMDP-T significantly outperforms ER in terms of patient~benefit and has slightly lower power. The bias and MSE for DP have not changed much in comparison with~$n=75,$ which is possibly due to the algorithm allocating all trial~participants to one treatment after a certain time point. The bias and MSE for ER, CRDP and CMDP-T have decreased significantly in comparison with Figure~\ref{results_N75}, where the same ordering as for~$n=75$ is seen in the OCs. Figure~\ref{results_N200} shows again that a trade-off was made between patient~benefit and OCs bias and MSE for the \mbox{CMDP-T~procedure.}

 \begin{figure}[htb]
		\includegraphics[width = \linewidth]{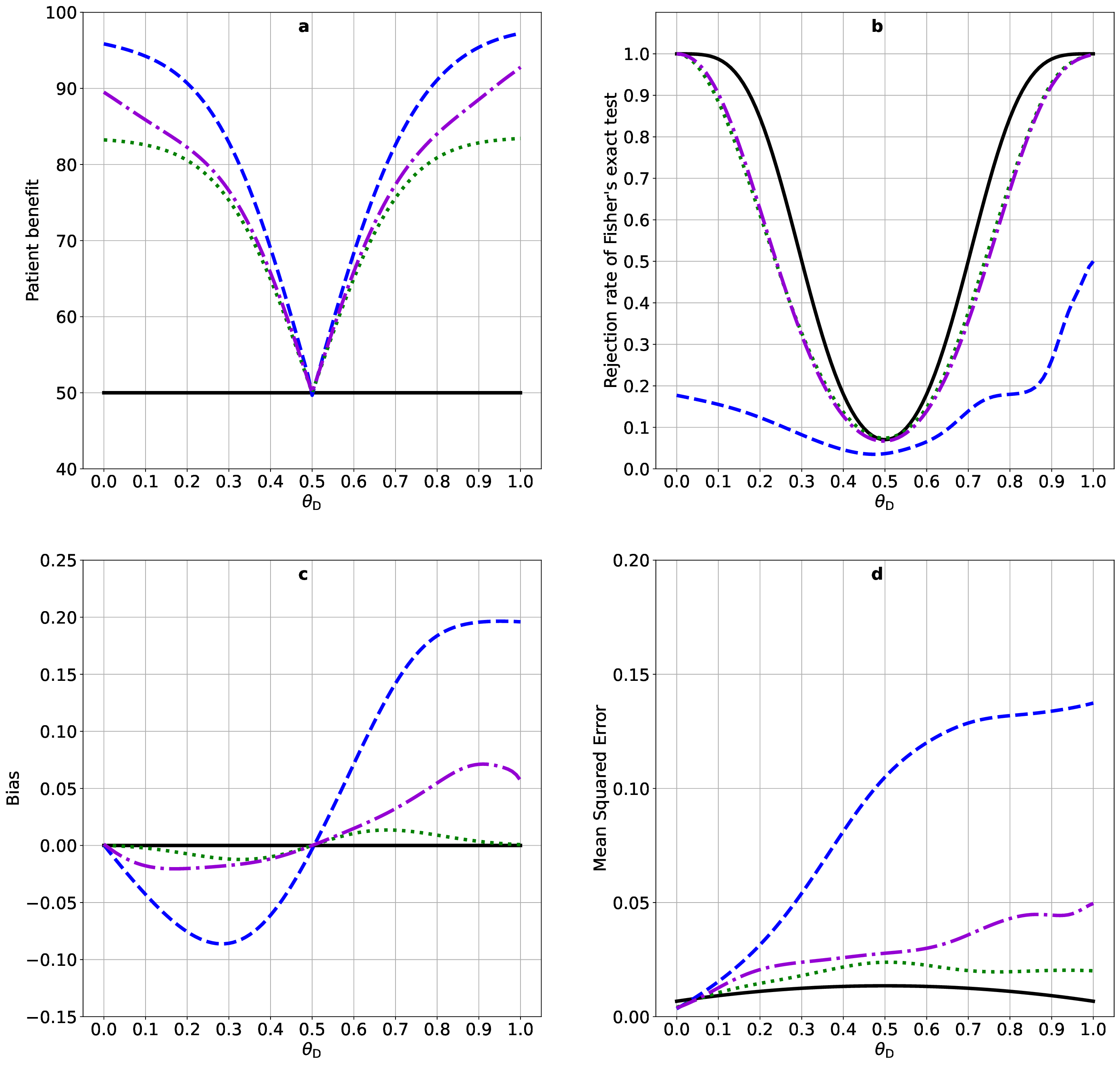}
	\caption{Patient benefit (subfigure~a), rejection rate (subfigure~b),  bias (subfigure~c), and mean squared error (subfigure~d) vs.~$\theta_\D$ for~$\theta_\C=0.5$,~$n=75$,~$\alpha^* = 0.05$,~$\beta= 0.4$, and RA~procedures ER (solid), DP (dashed), CRDP (dotted) and CMDP-T (dash-dotted)}\label{results_N75}
\end{figure}
\FloatBarrier

 \begin{figure}[htb]
		\includegraphics[width=\linewidth]{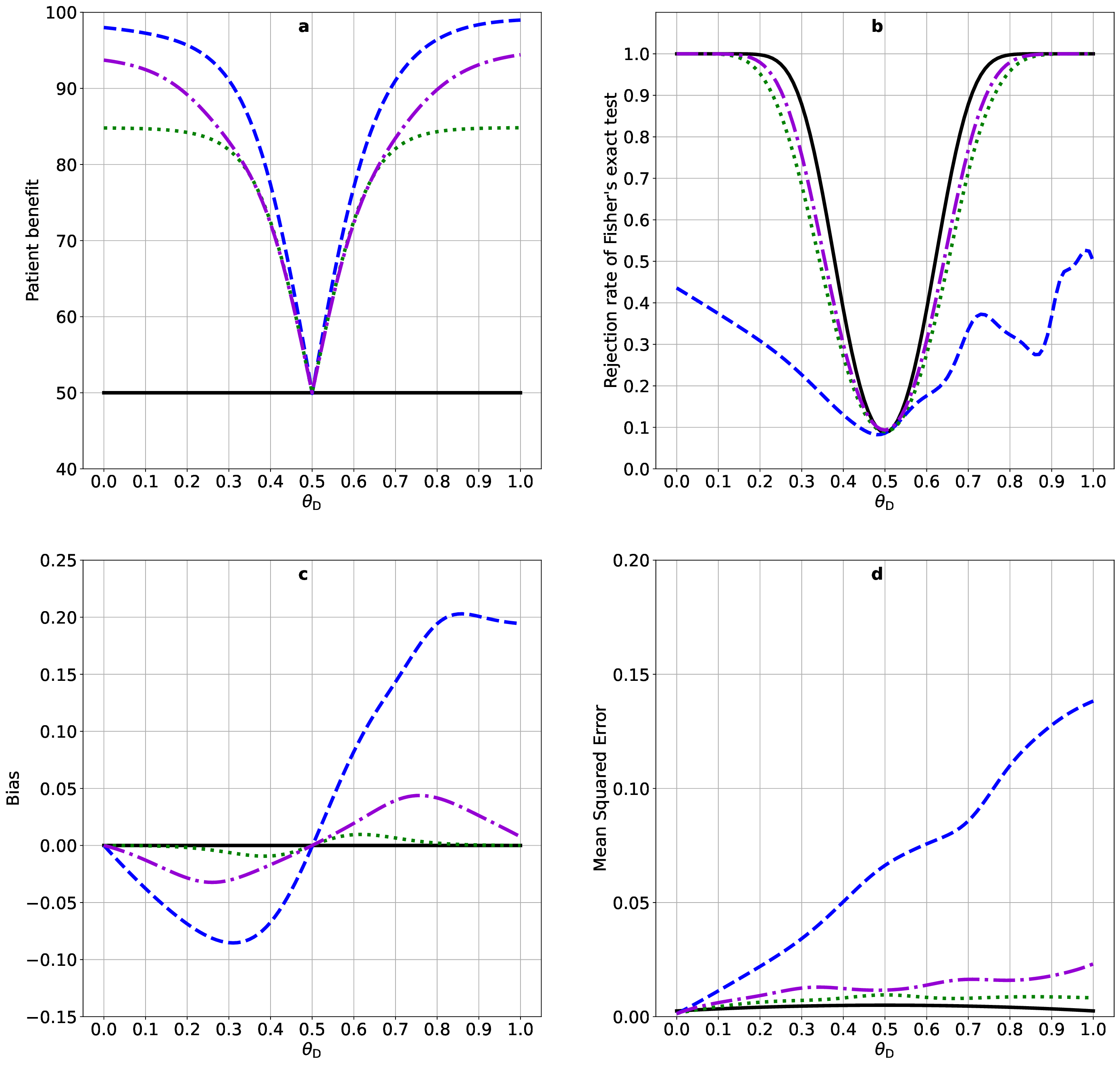}
	\caption{Patient benefit (subfigure~a), rejection rate (subfigure~b),  bias (subfigure~c), and mean squared error (subfigure~d) vs.~$\theta_\D$ for~$\theta_\C=0.5$,~$n=200$,~$\alpha^* = 0.07$,~$\beta= 0.23$, and RA~procedures ER (solid), DP (dashed), CRDP (dotted) and CMDP-T (dash-dotted)}\label{results_N200}
\end{figure}
\FloatBarrier

Figures~\ref{results_N75_det} and \ref{results_N200_det} in Appendix~\ref{app:plots_type_1_power} show the OCs according to the comparison described above when~$p=1.00$ for CMDP-T instead of~$0.95$, i.e., when a DRA CMDP-T procedure is used. 
Figures~\ref{results_N75_det} and \ref{results_N200_det} show that for the DRA CMDP-T procedure, the patient~benefit increases and the type~I~error and power remain under control, while the bias and MSE are worse than for CMDP-T with~$p=0.95$ and take on values similar to those for DP when~$\theta_\D\gg \theta_\C.$ Hence, the randomization incorporated in the CMDP-T procedure with~$p=0.95$ mitigates a large part of the bias and MSE. Figures~\ref{results_N75_25} - \ref{results_N200_75} in Appendix~\ref{app:plots_type_1_power} show the OCs according to the comparison above when~$p\in\{0.95,1.00\}$,~$n\in\{75,200\}$, for several values of~$\theta_\D$, and~$\theta_\C\in\{0.25,0.75\}$. The behaviour of the power remains similar when varying~$p$ and~$\theta_\C$. When~$\theta_\C=0.25$, a uniform outperformance of CRDP by CMDP-T is no longer seen in terms of power as the rejection rate for CRDP is higher for~$\theta_\D\approx0.25$.
Note, however, that the rejection rate for CRDP is only significantly higher than that for CMDP-T for parameter values close to the null, where the rejection rate is low anyway. 

\subsection{Application 2: Control of estimation error}
Two CMDP-E policies, calculated using~\eqref{CMDP_EST} in Example~\ref{example: CMDP-E} and denoted~\mbox{CMDP-E1} and~CMDP-E2, are evaluated. For policy CMDP-E1, the values of~$\alpha^*,\beta, V_{\sigma}$ are chosen such that the MSE is similar to that of ER, while there is a gain in patient~benefit. 
For \mbox{CMDP-E1} we set~$\mathcal{S}=\{[0,1]^2\}$, i.e.,~$\mathcal{S}$ only contains the unit square, and we calculate the value of~\eqref{eqn:MSE} realised under the ER RA procedure and set~$V_\sigma$ to~$\xi_1$ times this value for~$\xi_1\in[1,\infty)$, we \mbox{set~$\alpha^*=\beta=1$,} i.e., the type I and II error constraints are automatically satisfied, hence we focus solely on MSE for the policy \mbox{CMDP-E1.} 
For CMDP-E2 we take~$[\theta_a^\ell, \theta_a^u)\in\{[0,0.25),[0.25,0.5),[0.5,0.75),[0.75,0.9),[0.9,1.0)\}$ for all~$a\in\{\C,\D\}$ in order to construct~$\mathcal{S}$, dividing the unit square in blocks with a surface~$0.25^2$, as well as some additional blocks \mbox{for~$\theta_\C\geq 0.9$} \mbox{or~$\theta_\D\geq 0.9$} where the largest imbalances in treatment group sizes occur (for RA procedures inducing high patient~benefit). We set~$V_\sigma$ equal to~$\xi_2$ times the value of~\eqref{eqn:MSE} realised under CRDP for~$\xi_2\in[1,\infty).$

First, we evaluate the results for CMDP-E for~$n=75$.
After evaluating several choices based on the OCs, we set~$\xi_1 = 1.05$. The prior parameters~$\tilde{s}_{a,0}, \tilde{f}_{a,0}$ for~$\Pi_0$ are set to 1, as are the parameters for the prior~$\Pi$.
Computation of the \mbox{CMDP-E1} policy took 28.2 seconds, and resulted in a relative optimality gap of~$8.18\cdot 10^{-7}$ and total expected reward~41.3 under the prior~$\Pi$, while computation of \mbox{the CMDP-E2 policy}, where we set~$\xi_2=1.00,$  \mbox{$\alpha^*=0.05,$} \mbox{$\beta = 0.4$,} took~708~seconds, and resulted in a relative optimality gap of~$4.59\cdot 10^{-4}$ and total expected reward~45.4 under the prior~$\Pi$ in comparison to~$\timeCRDPSMALLHOR$ seconds for CRDP with total expected reward~$\rewCRDPSMALLHOR$.
 Figure~\ref{fig:results_estimation_n75} shows the results.  Figure~\ref{fig:results_estimation_n75} shows that the CMDP-E1~procedure induces higher patient~benefit than the ER RA procedure, while having similar power and MSE. The  \mbox{CMDP-E1} procedure induces higher patient~benefit when~$\theta_\D\leq 0.5$ in comparison to when~$\theta_\D\geq 0.5.$
 This behaviour could be explained by the fact that allocating more patients to the arm with the highest variance (in the spirit of Neyman allocation~\citep{rosenberger2004}), which is always the control arm in this evaluation, maximizes power and hence induces good OCs. 
Figure \ref{fig:results_estimation_n75} shows that \mbox{the~CMDP-E2~procedure} induces similar patient~benefit, power and MSE to CRDP, while showing slightly higher bias. This result is curious, as the two policies are found using two completely different procedures, but could be explained by the fact that the parameters~$\ell$ and~$p$ for CRDP are tuned specifically to balance patient~benefit, power, bias, and MSE, while the constraints for CMDP-E2 are based on the attained MSE for CRDP. 

Second, we evaluate the results for~$n=200$.
For CMDP-E1, we set~$\xi_1 = 1.1$. The prior parameters~$\tilde{s}_{a,0}, \tilde{f}_{a,0}$ for~$\Pi_0$ are set to 1, as are the parameters for the prior~$\Pi$.
Computation of the  \mbox{CMDP-E1} policy took 137 seconds, and resulted in a relative optimality gap of~$3.98\cdot 10^{-7}$ and total expected reward 113 under the prior~$\Pi$, while computation of the CMDP-E2 policy, where we set~$\xi_2=1.05,\;$\mbox{$\alpha^*=0.07$},~$\beta =0.753$, took 190$\cdot 10^2$ seconds, and resulted in a relative optimality gap of~$2.75\cdot 10^{-4}$ and total expected reward 123 under the prior~$\Pi$ in comparison to~$\timeCRDPhighHOR$ seconds for CRDP with total expected reward~$\rewCRDPhighHOR$.
 Figure \ref{fig:results_estimation_n200} shows the results.  Figure~\ref{fig:results_estimation_n200} shows that the \mbox{CMDP-E1}~procedure again induces higher patient~benefit than the~ER RA procedure while having similar power and MSE, while inducing higher patient~benefit \mbox{when~$\theta_\D\leq 0.5$} in comparison to \mbox{when~$\theta_\D\geq 0.5.$}
Figure~\ref{fig:results_estimation_n200} shows that the \mbox{CMDP-E2~procedure} induces similar power and MSE to CRDP, while showing slightly higher bias and slightly higher patient~benefit. 

In conclusion, for~$n=200$ the CMDP-E2~procedure   outperforms CRDP in terms of patient~benefit, while showing similar MSE and power and slightly higher bias, which can possibly be mitigated by bias reduction techniques~\citep{Bowden2017unbiased}. Such procedures can however also lead to higher variance of the treatment effect estimator. The outperformance of CRDP by the  \mbox{CMDP-E1} procedure in terms of patient benefit is not as high as in the previous application, indicating that CRDP is a well-performing policy for a vast range of parameter values and OCs, while a CMDP procedure such as CMDP-T may outperform CRDP for specific settings/OCs. 

Figures~\ref{results_N75_det_estimation} and \ref{results_N200_det_estimation} in Appendix~\ref{app:plots_type_1_power} show the OCs according to the comparison described above when~$p=1.00$  instead of~$0.95$ for  \mbox{CMDP-E1} and CMDP-E2, i.e., when a DRA version of  \mbox{CMDP-E1} and CMDP-E2 is used. 
Figures~\ref{results_N75_det_estimation} and \ref{results_N200_det_estimation} show very similar measures to the  \mbox{CMDP-E1} and CMDP-E2 RA procedures for~$p=0.95$, in contrast to the comparison with~$p=0.95$ vs.~$p=1.00$ for CMDP-T, where large differences in bias and MSE were seen. The RAR procedures  \mbox{CMDP-E1} and \mbox{CMDP-E2 ($p=0.95$)} are preferred over the DRA procedures  \mbox{CMDP-E1} and \mbox{CMDP-E2 ($p=1.00$)} as the patient~benefit is similar, while using an RAR policy brings advantages, e.g., in terms of selection bias mitigation. 
Figures~\ref{results_N75_25_estimation} - \ref{results_N200_75_estimation} in Appendix~\ref{app:plots_type_1_power} show the OCs according to the comparison above when~$p\in\{0.95,1.00\}$,~$n\in\{75,200\}$, and~$\theta_\C\in\{0.25,0.75\}$. 
Figures~\ref{results_N75_25_estimation} - \ref{results_N200_75_estimation} show that the difference in the power and patient benefit between CRDP and CMDP-E2 remains small when varying~$p$ and~$\theta_\C$, while larger differences are seen between the bias and MSE of CRDP and CMDP-E2 for~$\theta_C\in\{0.25,0.75\}$ and~$n=200$.

\FloatBarrier
 \begin{figure}[htb]
		\includegraphics[width=\linewidth]{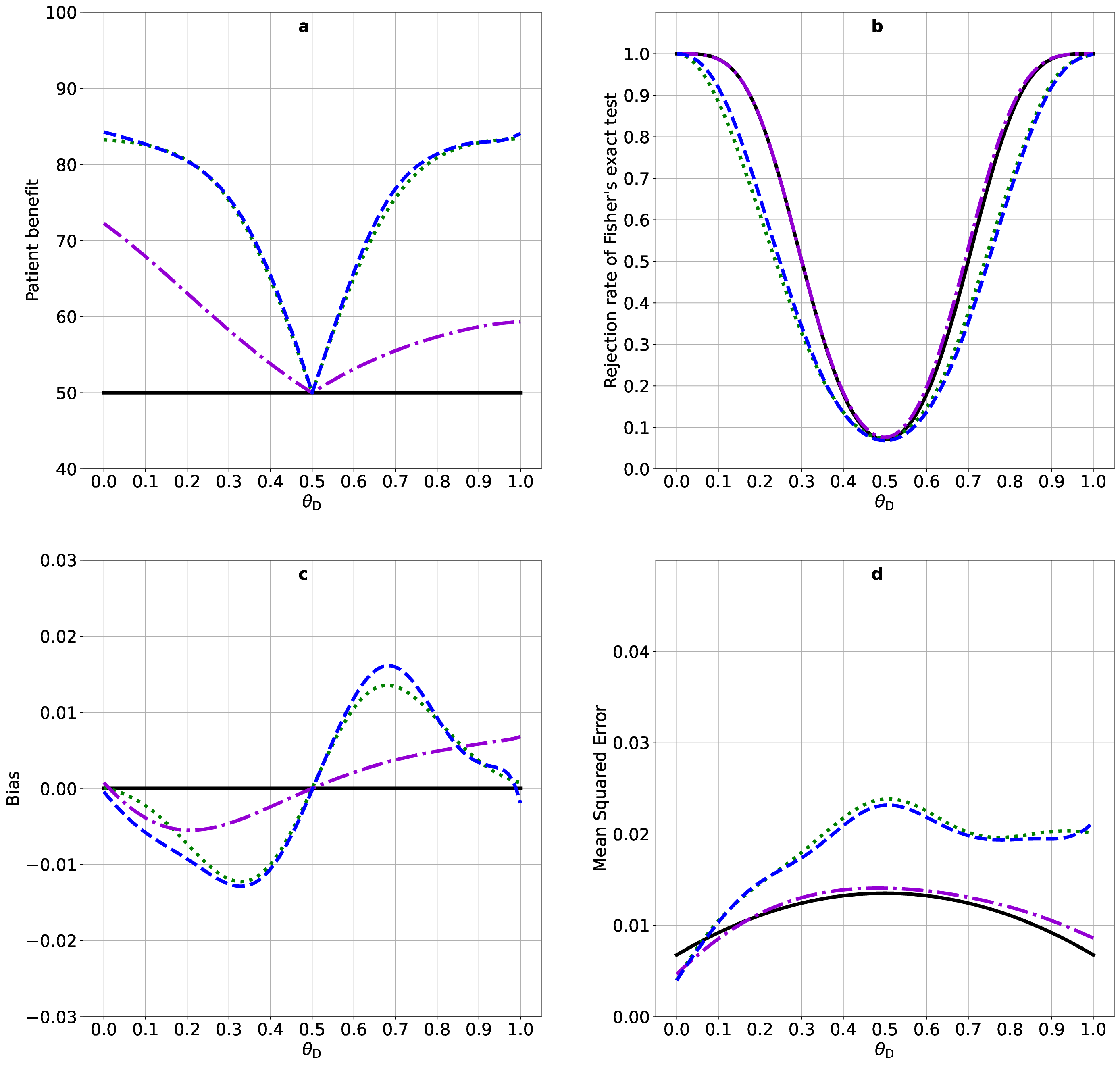}

	\caption{Patient benefit (subfigure~a), rejection rate (subfigure~b),  bias (subfigure~c), and mean squared error (subfigure~d) vs.~$\theta_\D$ for~$\theta_\C=0.5$,~$n=75$, and RA~procedures ER (solid), CMDP-E2 (dashed), CRDP (dotted) and  \mbox{CMDP-E1} (dash-dotted)}\label{fig:results_estimation_n75}
\end{figure}
\FloatBarrier
    
\FloatBarrier
 \begin{figure}[htb]
\includegraphics[width=\linewidth]{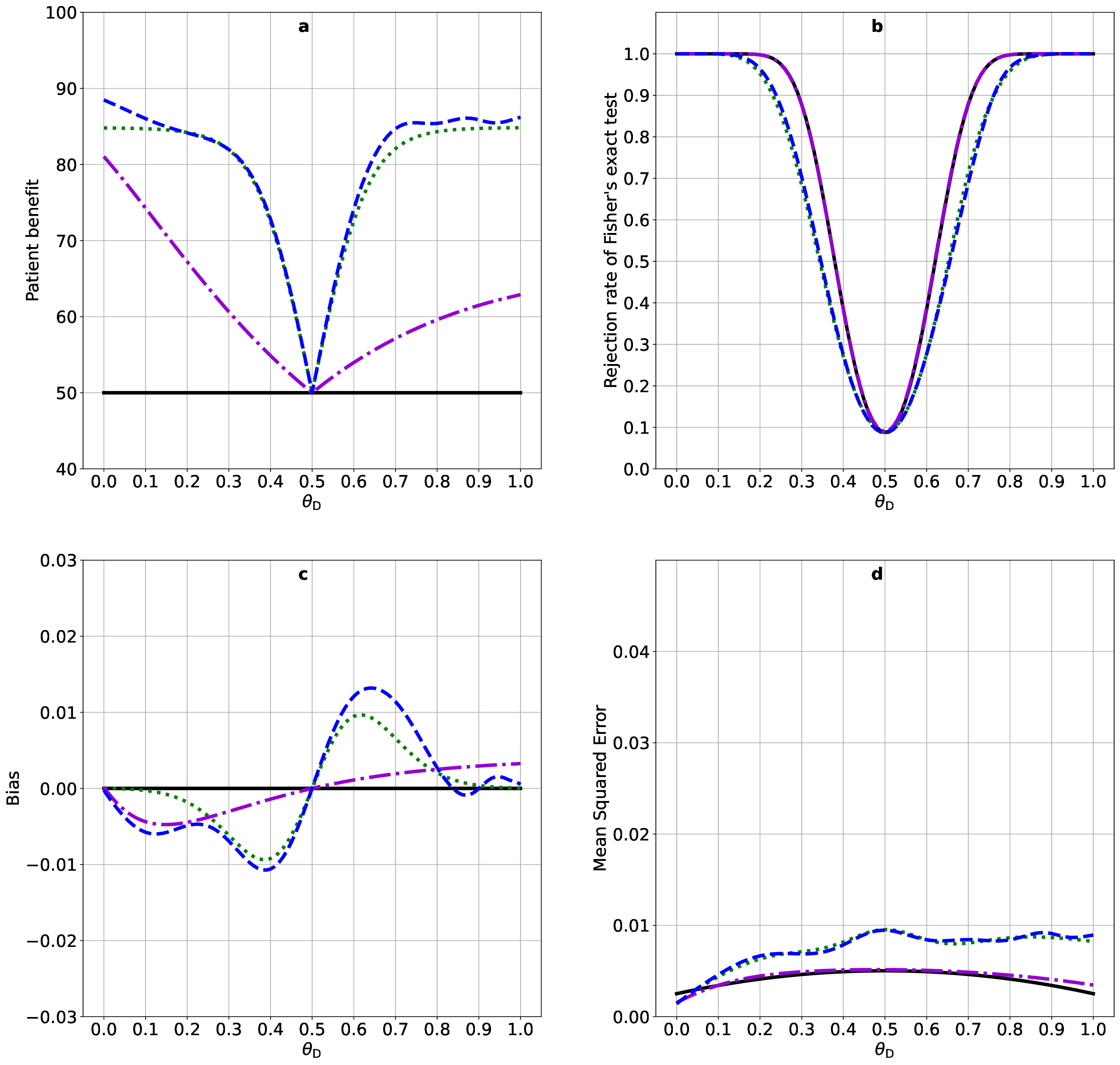}	
	\caption{Patient benefit (subfigure~a), rejection rate (subfigure~b),  bias (subfigure~c), and mean squared error (subfigure~d) vs.~$\theta_\D$ for~$\theta_\C=0.5$,~$n=200$, and RA~procedures ER (solid), CMDP-E2 (dashed), CRDP (dotted) and  \mbox{CMDP-E1} (dash-dotted)}\label{fig:results_estimation_n200}
\end{figure}
\FloatBarrier

\subsection{Application 3: Robustness to prior misspecification}
First, we evaluate CMDP-R (Example~\ref{example:CMDP-R}) for~$n=200$ and~$\tilde{s}_{\C,0}=3$,~$\tilde{f}_{\C,0}=7$, \mbox{$\tilde{s}_{\D,0}=6$},~$\tilde{f}_{\D,0}=4$,  hence under the informative prior the control and development treatment have a prior probability of success of 0.3 and 0.6, with a prior sample size of~10. The prior~$\Pi_\text{LI}$ corresponds to an independent uniform prior for both success probabilities. We do not consider CRDP in this section as this RA procedure was not constructed with sensitivity to prior misspecification in mind, hence, we will also not present results for~$n=75$ as in~\citet{williamson2017bayesian}. In fact, to the best of our knowledge, we are the first to introduce a BRA procedure focusing on robustness against prior misspecification.

To investigate different settings of robustness to the prior for the CMDP-R, we compare policies found under~\eqref{CMDP_ROBUST} for~$\xi = 0.00,\;0.990,\;0.999$ and~$1.00$ according to the OCs, which took 4.84, 123, 1260, and  3.34 seconds to compute, with optimality \mbox{gaps 0.00,~$4.87\cdot 10^{-3}$,~$4.10\cdot 10^{-3}$}, and 0.00, and optimal values 118,  117, 117, and 116. Figure \ref{results_priorrobust_ess10} shows the patient benefit, RR,  bias, and MSE for the four different CMDP-R policies, where we now set~$\theta_\C=0.3$ and again vary~$\theta_\D$ to investigate prior robustness. Figure \ref{results_priorrobust_ess10} shows that the risk under prior misspecification decreases with~$\xi$, as the patient benefit increases for~$\theta_\D<0.3$ in~$\xi.$ The ordering of patient benefit flips after~$\theta_\D$ passes the value~$0.3$, i.e., depending on whether the prior rightly \mbox{specifies~$\theta_\D\geq \theta_\C$,} and the
 patient benefit curve becomes symmetric \mbox{when~$\xi = 1.0.$} As a higher amount of allocations to the optimal treatment means a larger imbalance in treatment group sizes, the reverse ordering in quality, flipped after~$\theta_\D$ passes~$0.3$, is seen for the power and MSE, where this behaviour is seen to occur more drastically for the power, where the ordering flips exactly at~$\theta_\D=0.3.$  
 Comparing the bias curves, the curves become more symmetrical around~$0.3$ when~$\xi$ increases. This is because the robustness constraint enforces that the policy depends less on the informative prior when~$\xi$ increases, hence the observations tend to guide the policy more when~$\xi~$ is close to~$1,$ inducing a change in sign around~$\theta_\D=0.3$ and a symmetric behaviour for the policy in both arms, leading to a more symmetric looking bias.

Second, we evaluate the results for~$n=200$ and~$\tilde{s}_{\C,0}=30$,~$\tilde{f}_{\C,0}=70$, \mbox{$\tilde{s}_{\D,0}=60$},~$\tilde{f}_{\D,0}=40$,   hence under the informative prior the control and developmental treatment have a prior probability of success of 0.3 and 0.6, with a prior sample size of 100. The less informative prior was again chosen as a uniform prior for both success probabilities. 
We compare policies found for~\eqref{CMDP_ROBUST} for~$\xi = 0.00,\;0.900,\;0.990$ and~$1.00$, which took 3.00,  165,  171, and   2.95 seconds to compute, with optimality gaps~0.00,~$5.60\cdot 10^{-3}$,~$3.09\cdot 10^{-3}$, and~0.00, and optimal values 117,  117, 117, and~115.
 Figure \ref{results_priorrobust_ess100} shows the patient benefit, RR,  bias, and MSE for the four different CMDP-R policies, where we again set~$\theta_\C=0.3$ to investigate prior robustness. Figure \ref{results_priorrobust_ess100} again shows that the risk under prior misspecification decreases with~$\xi$, where the ordering of patient benefit flips after~$\theta_\D$ passes the value~$0.3$ and
 patient benefit curve becomes symmetric when~$\xi = 1.0$. The risk of prior misspecification is higher in this case, as the patient benefit stays around~$5\%$ even for~$\theta_\D=0.2,$ for~$\xi = 0.0$.  The flip in the ordering for~$\theta_\D<0.3$ vs.~$\theta_\D\geq 0.3$ is seen for the curves corresponding to~$\xi>0$, while the curve for~$\xi =0.0$ shows a different behaviour. Due to the large amount of prior certainty, the policy exploits knowledge on~$\theta_\D>\theta_\C$ for a large part of the parameter space, effectively acting as a fixed randomization policy with a probability of~$95\%~$ of allocating to the control arm. This is why there is low bias and roughly constant MSE for~$\theta_\D>0.3$. As the treatment groups are extremely unbalanced, the power is low and the MSE is high for such values of~$\theta_\D.$ As exploration is seen to kick in for CMDP-R with~$\xi=0.0$ when~$\theta_\D\ll0.3$, the treatment groups become more balanced when~$\theta_\D$ is close to zero, resulting in an increase in power. 

Note that, rewriting~\eqref{Eqn: minmaxV}, we have
$$L(\lambda) = \max_{\pi}\mathbb{E}^\pi\left[\sum_{t=0}^{n}{r}(\bX_t, \delta_t)\right] + \lambda \mathbb{E}^\pi_\text{LI}\left[\sum_{t=0}^{n}{r}_\text{LI}(\bX_t, \delta_t)\right] - \lambda v_\text{LI}\xi.~$$
The two expectations above can be rewritten to the expected sum of successes under a mixture prior of~$\Pi$ and~$\Pi_\text{LI}$, so an alternative method would be to formulate a mixture prior and optimise the sum of rewards under this mixture prior. The difference with the CMDP-R approach is that~\eqref{CMDP_ROBUST} gives a robustness guarantee of the resulting policy, something that might be of value when designing a clinical trial. Furthermore, one might easily increase the number of constraints in~\eqref{CMDP_ROBUST}, while adding more priors to the mixture would quickly make the procedure intractable as all weights need to be elicited beforehand.
 
\FloatBarrier
 \begin{figure}[htb]
		\includegraphics[width=\linewidth]{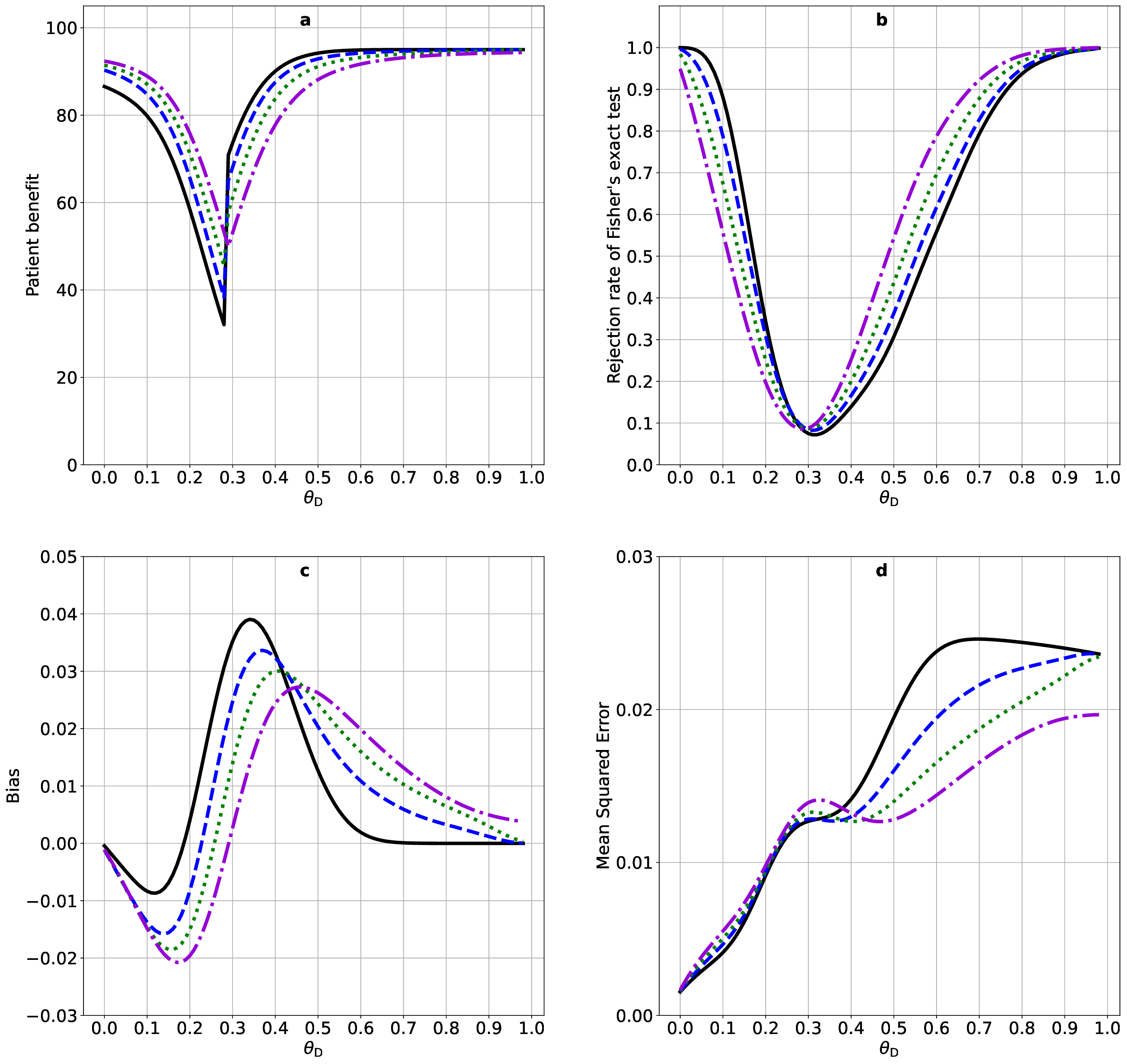}
		
	\caption{Patient benefit (subfigure~a), rejection rate (subfigure~b),  bias (subfigure~c), and mean squared error (subfigure~d) vs.~$\theta_\D$ for~$\theta_\C=0.3$,~$n=200$,~$\tilde{s}_{1,0}=3$,~$\tilde{f}_{1,0}=7$,~$\tilde{s}_{2,0}=6$,~$\tilde{f}_{2,0}=4$ and the \mbox{CMDP-R procedure}~\eqref{CMDP_ROBUST} with~$\xi=0.00,0.990,0.999,1.00$ denoted by the solid, dashed, dotted and dash-dotted lines respectively. }\label{results_priorrobust_ess10}
\end{figure}

 \begin{figure}[htb]
		\includegraphics[width=\linewidth]{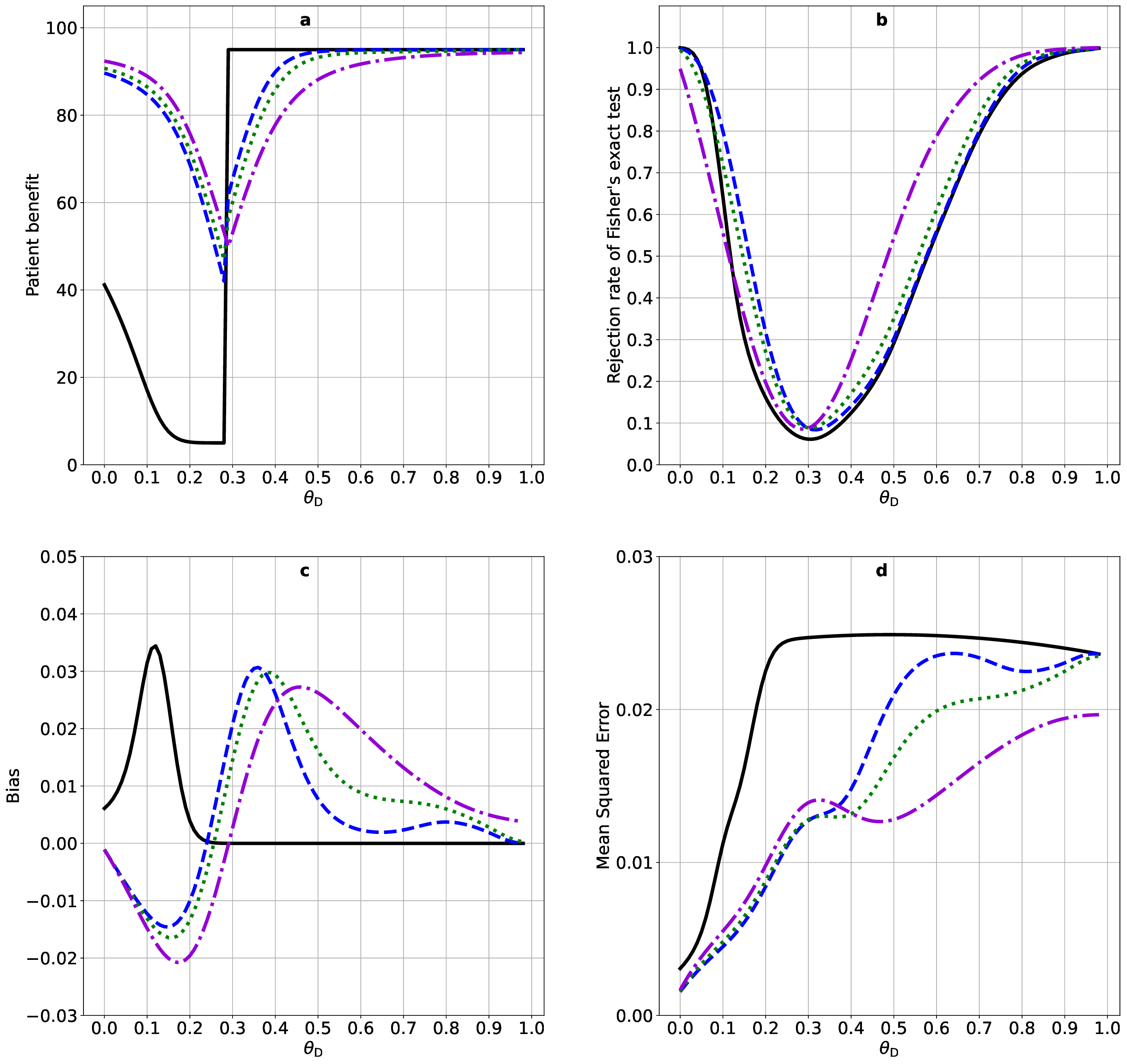}

	\caption{Patient benefit (subfigure~a), rejection rate (subfigure~b),  bias (subfigure~c), and mean squared error (subfigure~d) vs.~$\theta_\D$ for~$\theta_\C=0.3$,~$n=200$,~$\tilde{s}_{1,0}=30$,~$\tilde{f}_{1,0}=70$,~$\tilde{s}_{2,0}=60$,~$\tilde{f}_{2,0}=40$ and the \mbox{CMDP-R procedure}~\eqref{CMDP_ROBUST} with~$\xi=0.00,0.990,0.999,1.00$ denoted by the solid, dashed, dotted and dash-dotted lines respectively. }\label{results_priorrobust_ess100}
\end{figure}
\FloatBarrier

\section{Discussion }\label{sect:discussion}
The current paper introduced the class of constrained Markov decision process response-adaptive procedures. The constraints for the Markov decision process facilitate the construction of response-adaptive procedures with good OCs while prioritizing one operating characteristic in the trial (e.g., patient~benefit) by putting it in the objective. 

Three applications were presented in the current paper. In the first application, type~I~error and power constraints were formulated, and the resulting randomized response-adaptive procedure, CMDP-T, was shown to yield comparable power to the constrained randomized dynamic programming RA procedure introduced in~\citet{williamson2017bayesian} while inducing higher patient~benefit. Furthermore, it was shown that the deterministic version of  CMDP-T  also has the desired power properties, while inducing higher patient~benefit than the randomized CMDP-T procedure and constrained randomized dynamic programming. A deterministic procedure might however have certain non-desirable properties, e.g., allocation biases, but it is nonetheless curious that a deterministic policy with good power properties could be constructed.
 A topic of future research for CMDP-T could be to evaluate the performance of CMDP-T  for other tests such as tests for superiority or non-inferiority.

In the second application, we constructed two constrained Markov decision process procedures by formulating additional constraints on the MSE. The resulting procedures,  \mbox{CMDP-E1} and CMDP-E2, showed low MSE in addition to a similar power and type I error behaviour as seen for CMDP-T. Here, the priors under the constraints were designed in such a way that the constraint reflected the average behaviour of the policy on a specific part of the parameter space. 
The  \mbox{CMDP-E1} procedure showed higher patient~benefit than the non-adaptive equal allocation procedure while having similar MSE and power. The CMDP-E2 procedure showed OCs similar to constrained randomized dynamic programming, where only for a trial horizon of 200 participants, slight outperformance in patient~benefit by the constrained approach was seen, indicating the good performance and general applicability of the constrained randomized dynamic programming procedure.

In the third application, a constrained Markov decision process procedure, CMDP-R, was developed for prior robustness, where the objective is to maximize patient~benefit under an informative prior, with the restriction that the patient~benefit should also be sufficiently high under an uninformative prior on the success probabilities. It was shown that the constraint parameter~$\xi$ could be tuned such that the desired robustness to misspecification of prior information could be attained. To the best of our knowledge, ours is the first approach to incorporate robustness to prior misspecification directly into the design of the response-adaptive procedure.

Other applications might also be of interest to explore. For instance, in line with recent literature on higher moments and tail bounds of the regret distribution for multi-armed bandits~\citep{fan2023fragility}, it could be possible to define constraints on moments and tails of the number of suboptimal decisions made by the Markov decision process procedure, ensuring that the performance of the response-adaptive procedure in terms of patient~benefit is consistent over different sampled data sets. The constraints could also be used to constrain the allocations to both treatment groups to a minimum amount with high probability, as is the purpose of the penalty term in constrained randomized dynamic programming, and differences in OCs between such a constrained approach and constrained randomized dynamic programming could be investigated. 

In application 2 (where the focus was on the mean squared error), it was seen that in order to obtain a good performance in all OCs,  the parameters in the constraints had to be chosen in a specific manner. This is mainly because the constraints are Bayesian, hence they only guarantee a certain behaviour when averaged over a prior. In order to obtain good frequentist OCs in a more straightforward or automatic manner, it might be of interest to explore approaches more in line with robust optimisation, or to add a second layer of optimisation, guaranteeing good OCs for all possible parameter configurations. 

	In this paper, we proposed an algorithm based on Lagrange multipliers and backward recursion to calculate CMDP RA procedures. In our applications, we found that this algorithm yields a policy with a small relative optimality gap in all cases. However, there is no guarantee that this will always be the case. Namely, in settings where the amount of constraints is larger in comparison to the number of state-action pairs, this method might result in a highly suboptimal policy. Future research could focus on alternative solution methods that approximate the optimal solution to the CMPD problem in different ways, e.g., by using the linear programming formulation more directly.

For the sake of simplicity and tractability, the current paper considers a trial with two arms and binary outcomes. It would be interesting to extend the approach to more general settings, such as multiple trial stages, outcomes, and arms. 
As the computational effort for calculating constrained Markov decision process policies would quickly increase in such settings, it would be interesting to investigate whether approaches such as approximate dynamic programming or reinforcement learning could be (successfully) applied. 
It would furthermore be interesting to consider whether the constrained Markov decision process framework and results such as Lemma~\ref{lemma:datalikelihood_Bernoulli_RA} could be extended to different outcome types. One could think of a generalization of the model in Section~\ref{sect:model} along the lines in~\citet{YI2023125}, where a clinical trial with general outcome types was described using a Markov decision process. Furthermore, the assumption underlying the applicability of response-adaptive procedures is that the inter-arrival times of trial~participants are longer than the follow-up times. If this is not the case, the constrained Markov decision process needs to be able to deal with delayed responses. In~\citet{williamson2022generalisations}, the robustness of constrained randomized dynamic programming with respect to delayed responses was demonstrated. It would be interesting to investigate whether the same result can be shown for the constrained Markov decision process procedure. Lastly, as blocked RA procedures are often preferred in clinical trials~\citep{MERRELL2023107599},  it would furthermore be interesting to incorporate blocked allocation in the constrained Markov decision approach. 

The current paper provides theoretical results  and a method for finding an optimal policy for 
a class of finite-horizon constrained Markov decision processes that is more general than usually considered in the literature, as a different expectation operator, corresponding to a different prior belief, can be used in each constraint. Such optimisation problems arise naturally in settings where multiple hypotheses are considered for the treatment effect, e.g., when considering type~I~error control and power constraints. It is possible to also apply Lemma~\ref{lemma:datalikelihood_Bernoulli_RA} to multi-objective Markov decision processes, where each objective is formulated under a different prior.
In principle, Markov decision process response-adaptive procedures follow from a partially observable Markov decision framework, and hence it would be interesting to investigate whether there are other settings involving partial observability where the solution approach can be applied.

\FloatBarrier
\bibliographystyle{unsrtnat}

\bibliography{CMDP_for_RA_binary_arxiv}
\newpage
\appendix

\section{Additional plots}\label{app:plots_type_1_power}
\subsection{Application 1: Control of power and type~I~error}

 \begin{figure}[htb]
		\includegraphics[width=\linewidth]{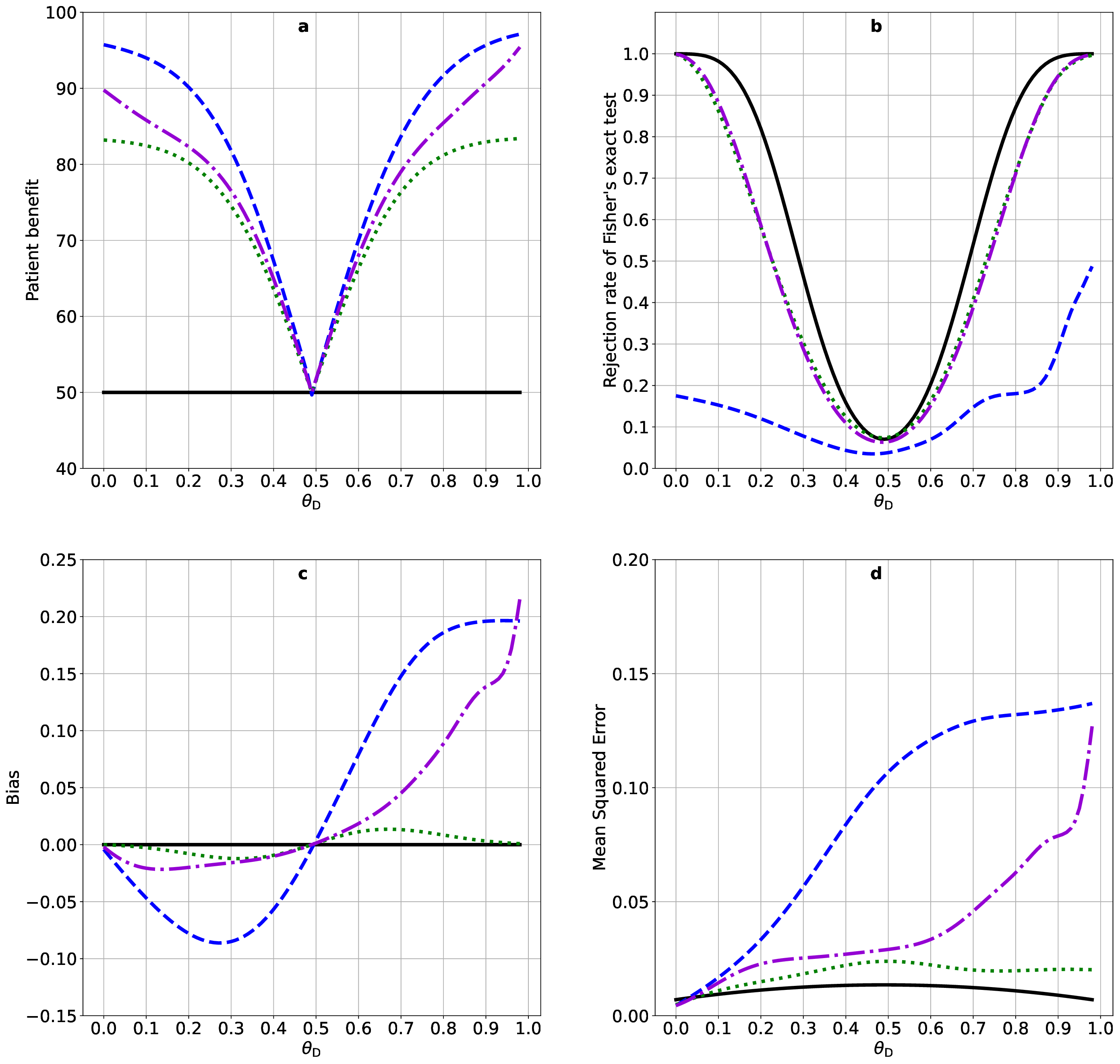}
		
	\caption{Patient benefit (subfigure~a), rejection rate (subfigure~b),  bias (subfigure~c), and mean squared error (subfigure~d) vs.~$\theta_\D$ for~$\theta_\C=0.5$,~$n=75$,~$\alpha^* = 0.05$,~$\beta= 0.4$, and RA~procedures ER (solid), DP (dashed), CRDP (dotted) and CMDP-T  with~$p=1.00$ (dash-dotted)  }\label{results_N75_det}
\end{figure}

 \begin{figure}[htb]
		\includegraphics[width=\linewidth]{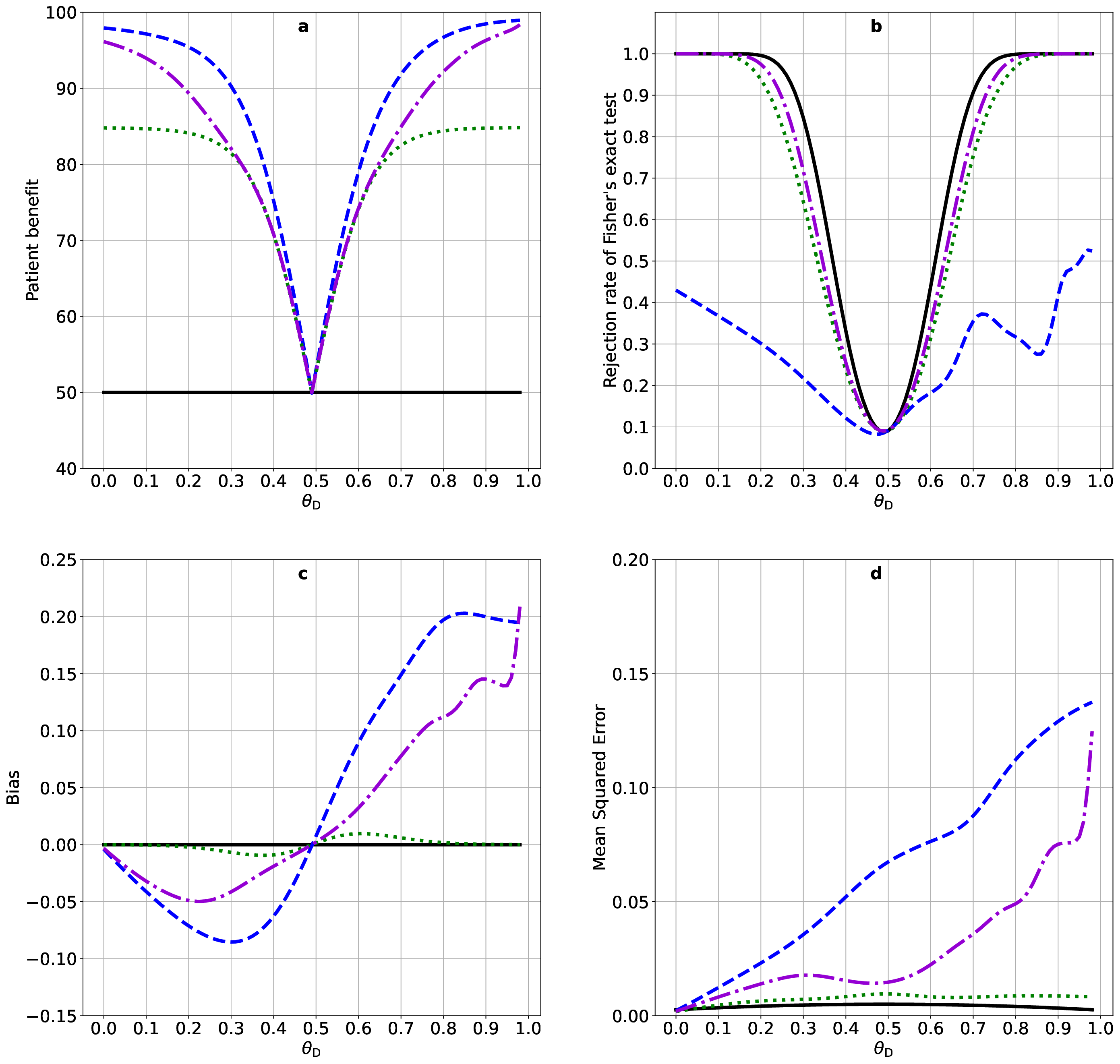}
		
	\caption{Patient benefit (subfigure~a), rejection rate (subfigure~b),  bias (subfigure~c), and mean squared error (subfigure~d) vs.~$\theta_\D$ for~$\theta_\C=0.5$,~$n=200$,~$\alpha^* = 0.07$,~$\beta= 0.23$, and RA~procedures ER (solid), DP (dashed), CRDP (dotted) and CMDP-T  with~$p=1.00$ (dash-dotted)  }\label{results_N200_det}
\end{figure}

 \begin{figure}[htb]
		\includegraphics[width=\linewidth]{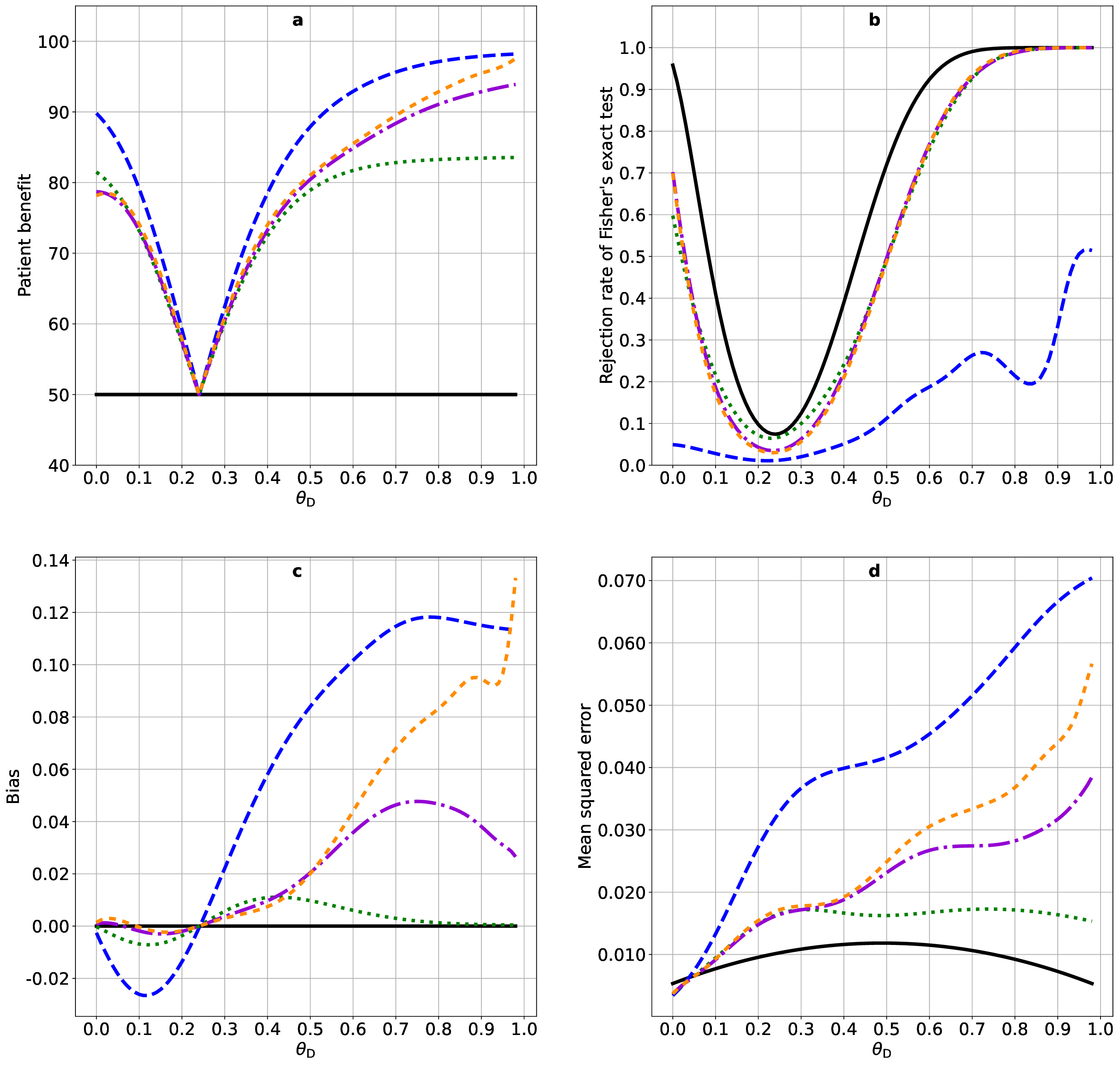}
		
	\caption{Patient benefit (subfigure~a), rejection rate (subfigure~b),  bias (subfigure~c), and mean squared error (subfigure~d) vs.~$\theta_\D$ for~$\theta_\C=0.25$,~$n=75$,~$\alpha^* = 0.05$,~$\beta= 0.4$, and RA~procedures ER (solid), DP (dashed, long), CRDP (dotted), CMDP-T with~$p=0.95$ (dash-dotted),  CMDP-T with~$p=1.00$  (dashed, short)}\label{results_N75_25}
\end{figure}

 \begin{figure}[htb]
		\includegraphics[width=\linewidth]{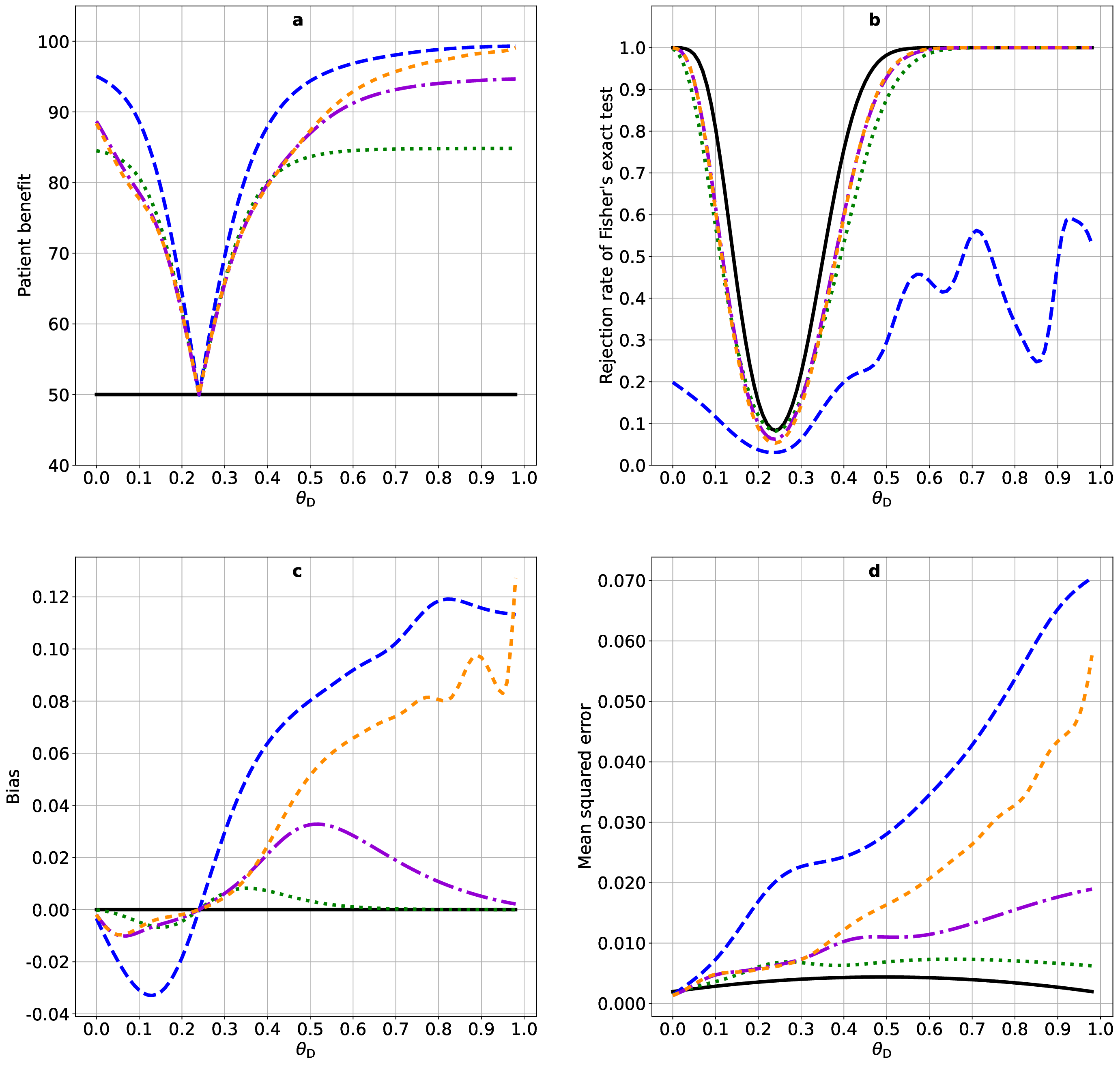}
		
	\caption{Patient benefit (subfigure~a), rejection rate (subfigure~b),  bias (subfigure~c), and mean squared error (subfigure~d) vs.~$\theta_\D$ for~$\theta_\C=0.25$,~$n=200$,~$\alpha^* = 0.07$,~$\beta= 0.23$, and RA~procedures ER (solid), DP (dashed, long), CRDP (dotted), CMDP-T with~$p=0.95$ (dash-dotted),  CMDP-T with~$p=1.00$  (dashed, short)}\label{results_N200_25}
\end{figure}

 \begin{figure}[htb]
		\includegraphics[width=\linewidth]{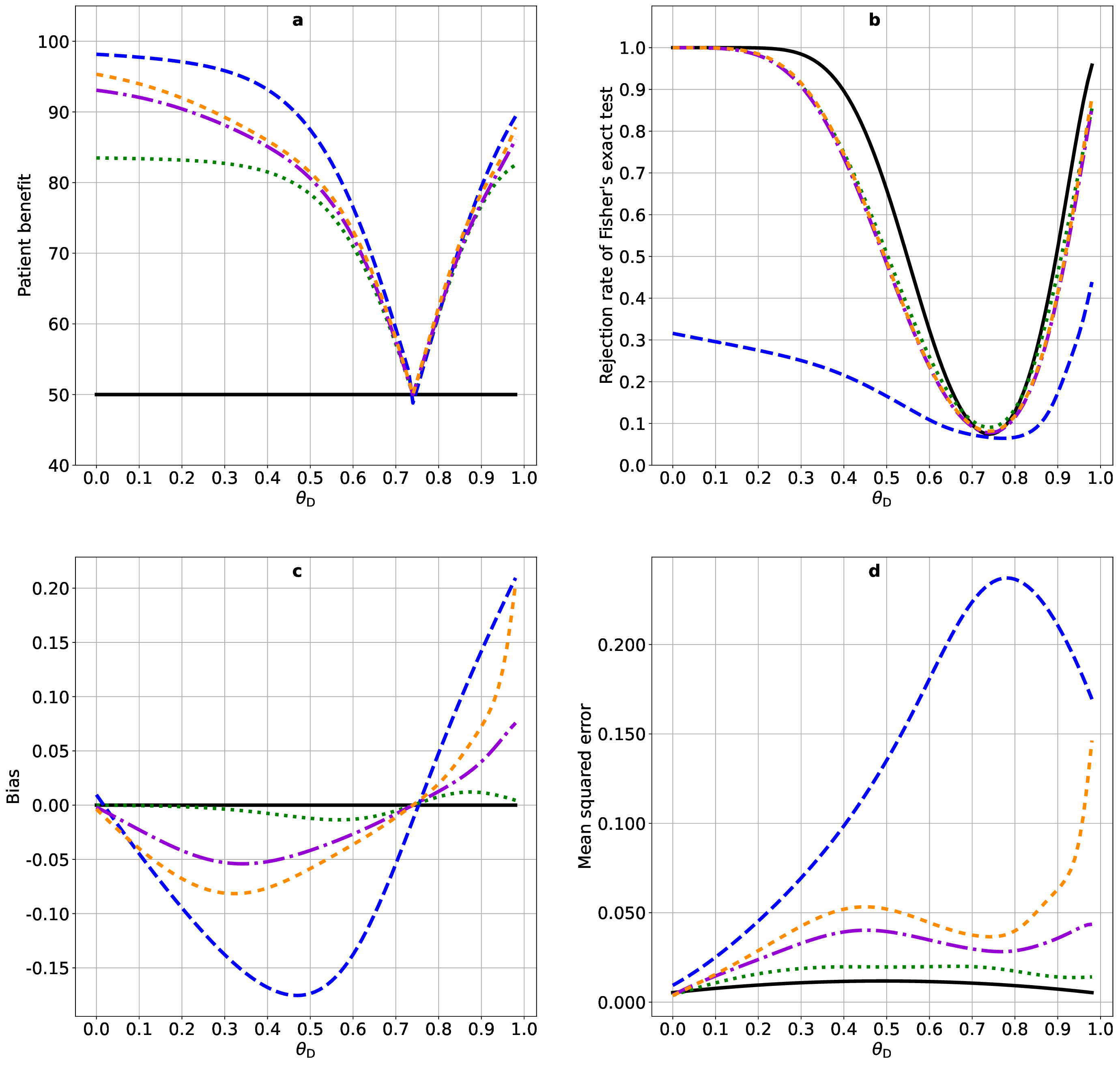}
		
	\caption{Patient benefit (subfigure~a), rejection rate (subfigure~b),  bias (subfigure~c), and mean squared error (subfigure~d) vs.~$\theta_\D$ for~$\theta_\C=0.75$,~$n=75$,~$\alpha^* = 0.05$,~$\beta= 0.4$, and RA~procedures ER (solid), DP (dashed, long), CRDP (dotted), CMDP-T with~$p=0.95$ (dash-dotted),  CMDP-T with~$p=1.00$  (dashed, short)}\label{results_N75_75}
\end{figure}

 \begin{figure}[htb]
		\includegraphics[width=\linewidth]{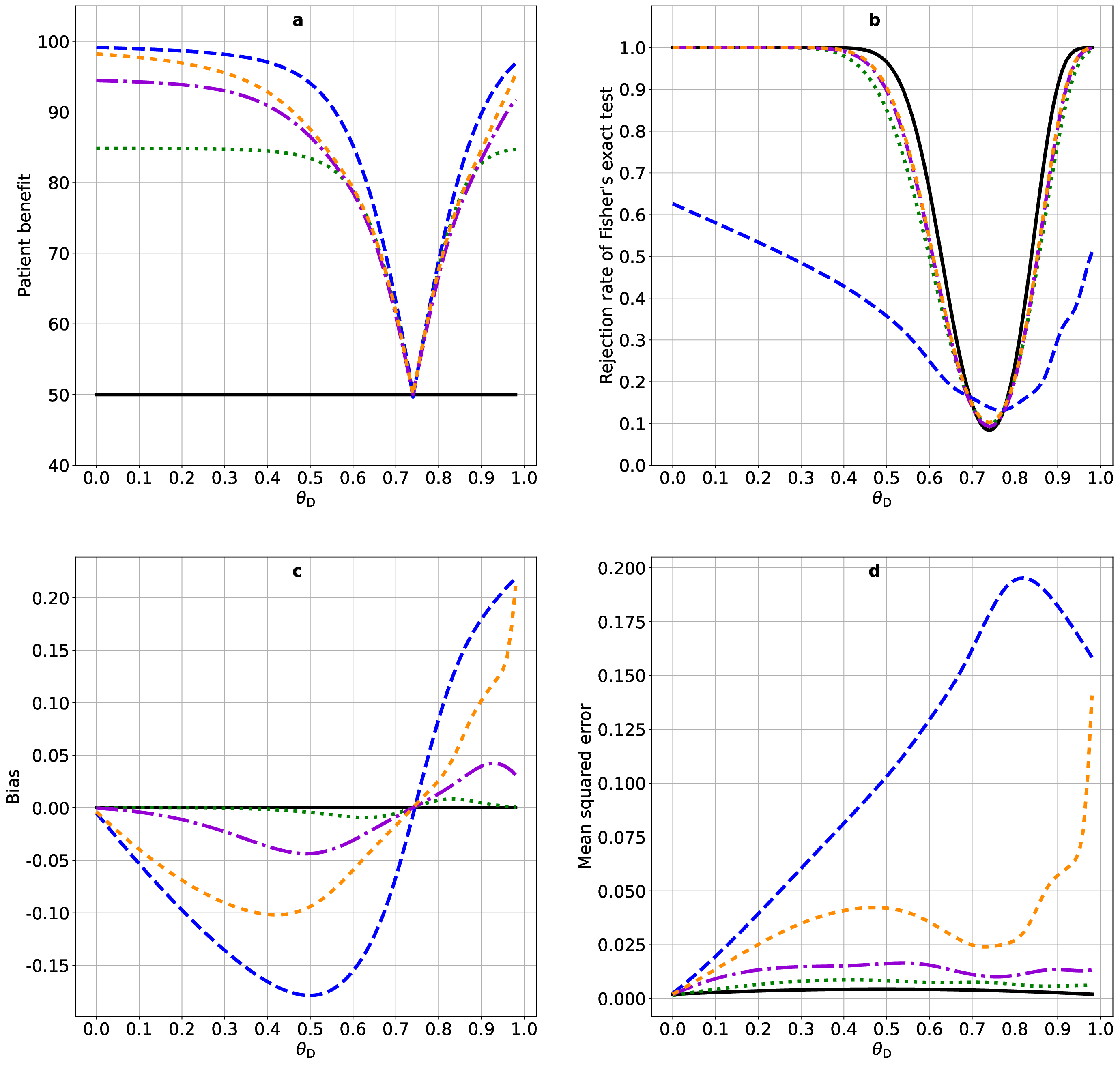}
	\caption{Patient benefit (subfigure~a), rejection rate (subfigure~b),  bias (subfigure~c), and mean squared error (subfigure~d) vs.~$\theta_\D$ for~$\theta_\C=0.75$,~$n=200$,~$\alpha^* = 0.07$,~$\beta= 0.23$, and RA~procedures ER (solid), DP (dashed, long), CRDP (dotted), CMDP-T with~$p=0.95$ (dash-dotted),  CMDP-T with~$p=1.00$  (dashed, short)}\label{results_N200_75}
\end{figure}
\FloatBarrier
\subsection{Application 2: Control of estimation error}
 \begin{figure}[htb]
		\includegraphics[width=\linewidth]{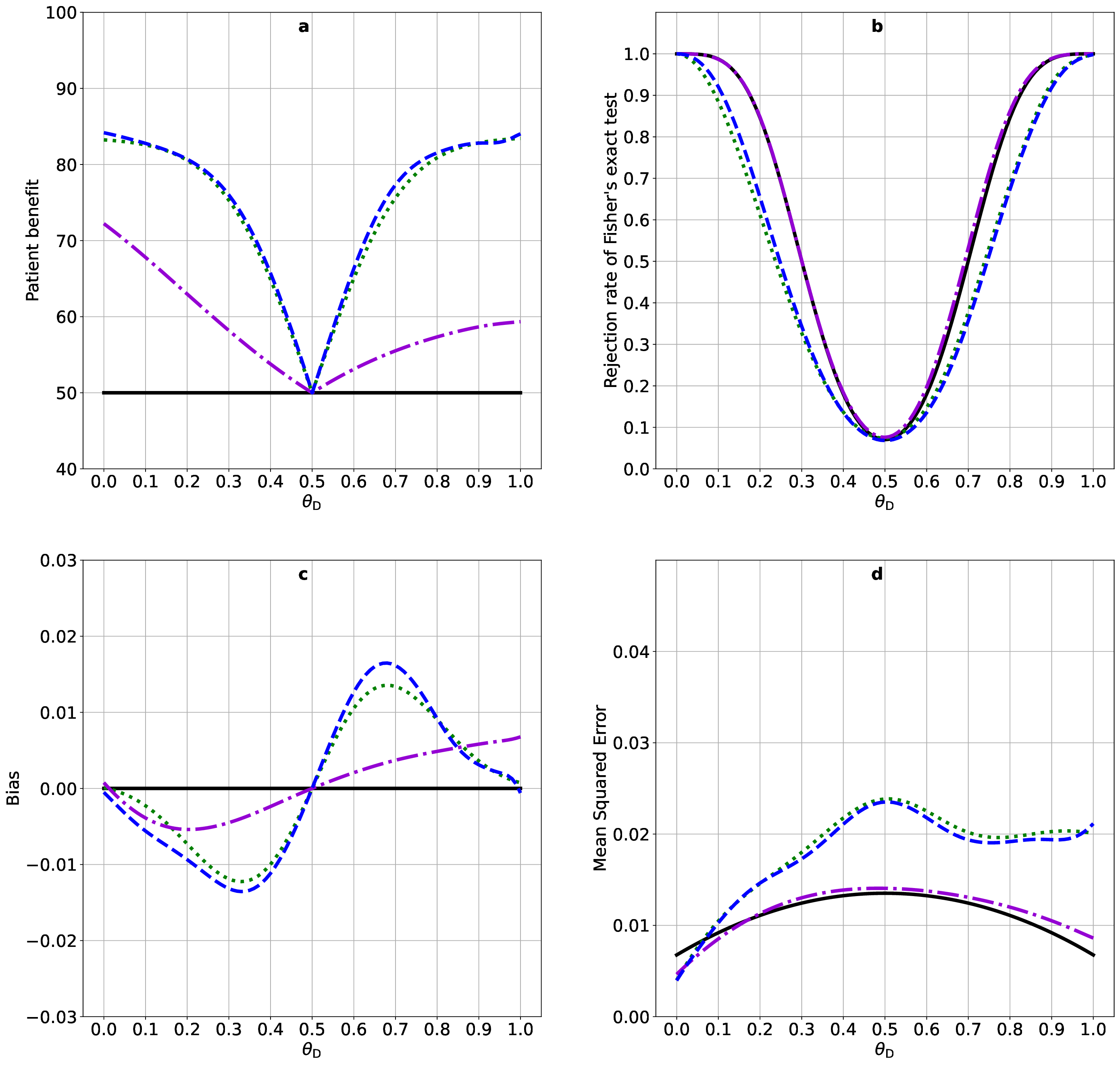}
		
	\caption{Patient benefit (subfigure~a), rejection rate (subfigure~b),  bias (subfigure~c), and mean squared error (subfigure~d) vs.~$\theta_\D$ for~$\theta_\C=0.5$,~$n=75$, and RA~procedures ER (solid), CMDP-E2 
with~$p=1.00$ (dashed), CRDP (dotted) and  \mbox{CMDP-E1} 
 with~$p=1.00$ (dash-dotted)}\label{results_N75_det_estimation}
\end{figure}

 \begin{figure}[htb]
		\includegraphics[width=\linewidth]{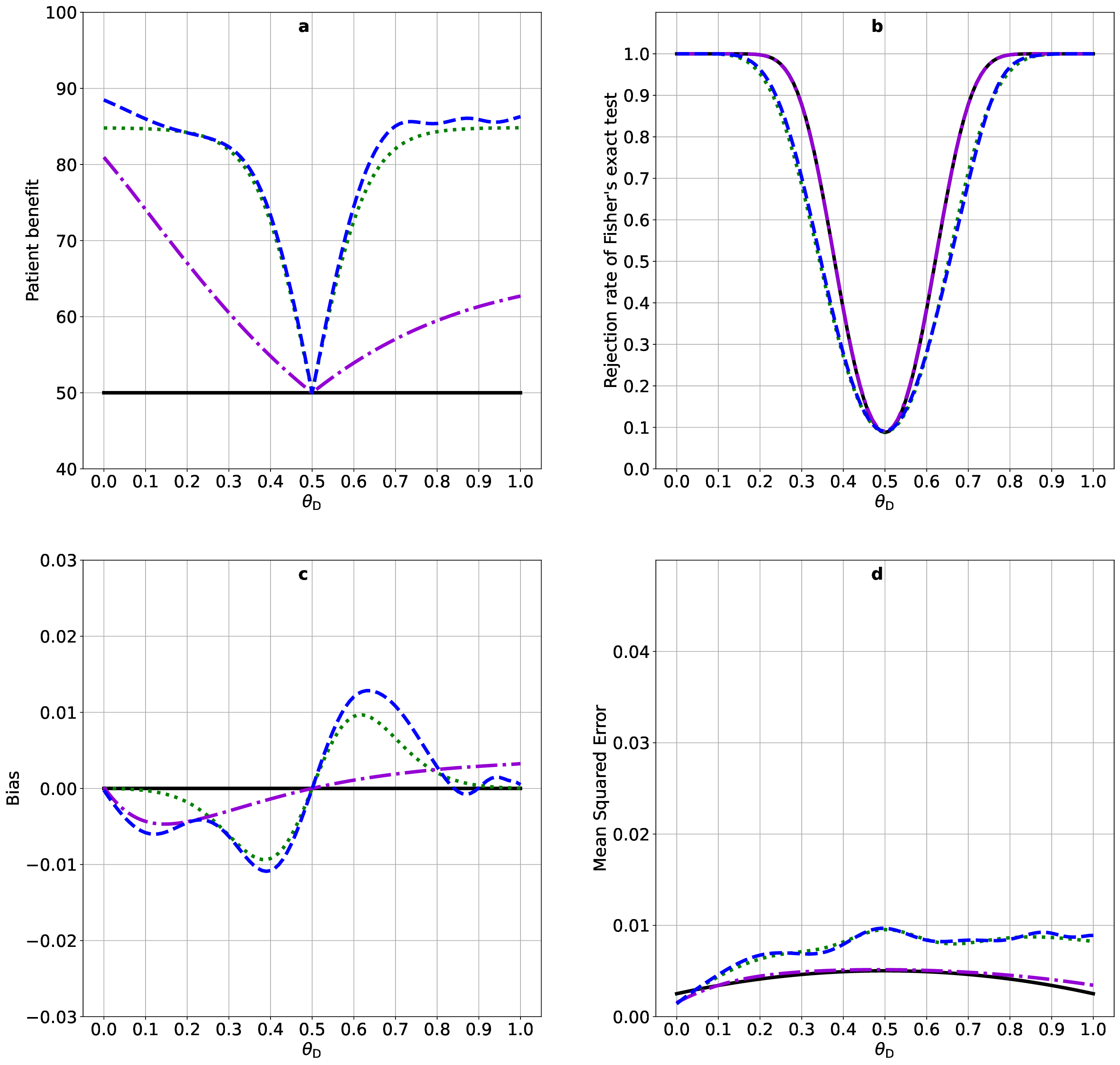}
	
	\caption{Patient benefit (subfigure~a), rejection rate (subfigure~b),  bias (subfigure~c), and mean squared error (subfigure~d) vs.~$\theta_\D$ for~$\theta_\C=0.5$,~$n=200$, and RA~procedures ER (solid), CMDP-E2 
with~$p=1.00$ (dashed), CRDP (dotted) and  \mbox{CMDP-E1} 
 with~$p=1.00$ (dash-dotted)}\label{results_N200_det_estimation}
\end{figure}

 \begin{figure}[htb]
		\includegraphics[width=\linewidth]{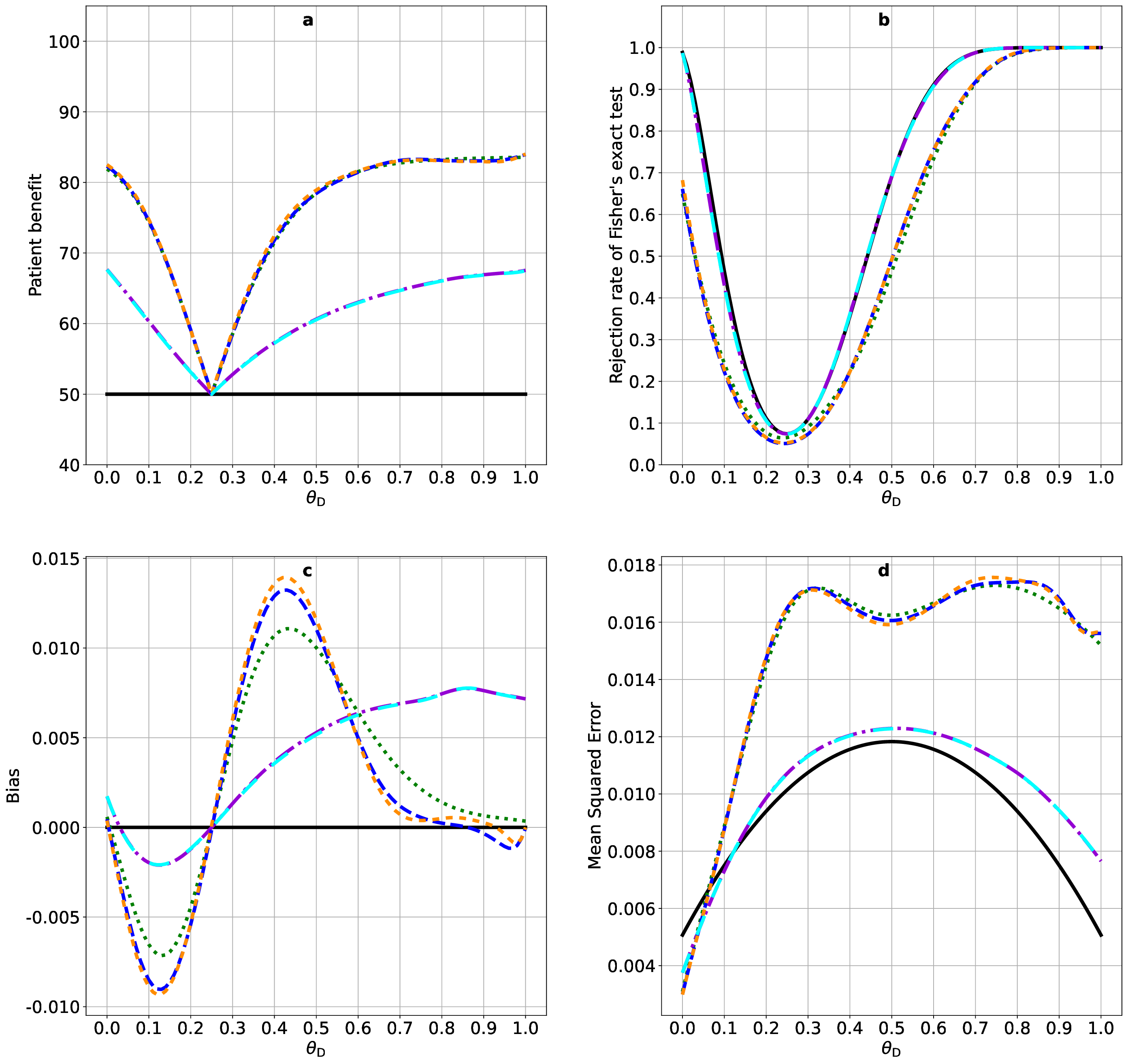}
	
	\caption{Patient benefit (subfigure~a), rejection rate (subfigure~b),  bias (subfigure~c), and mean squared error (subfigure~d) vs.~$\theta_\D$ for~$\theta_\C=0.25$,~$n=75$, and RA~procedures ER (solid), CMDP-E2 
with~$p=0.95$ (dashed, moderate), CMDP-E2 
with~$p=1.00$ (dashed, short), CRDP (dotted),  \mbox{CMDP-E1} 
with~$p=0.95$ (dash dot), CMDP-E2 
with~$p=1.00$ (dashed, long)}\label{results_N75_25_estimation}
\end{figure}

 \begin{figure}[htb]
		\includegraphics[width=\linewidth]{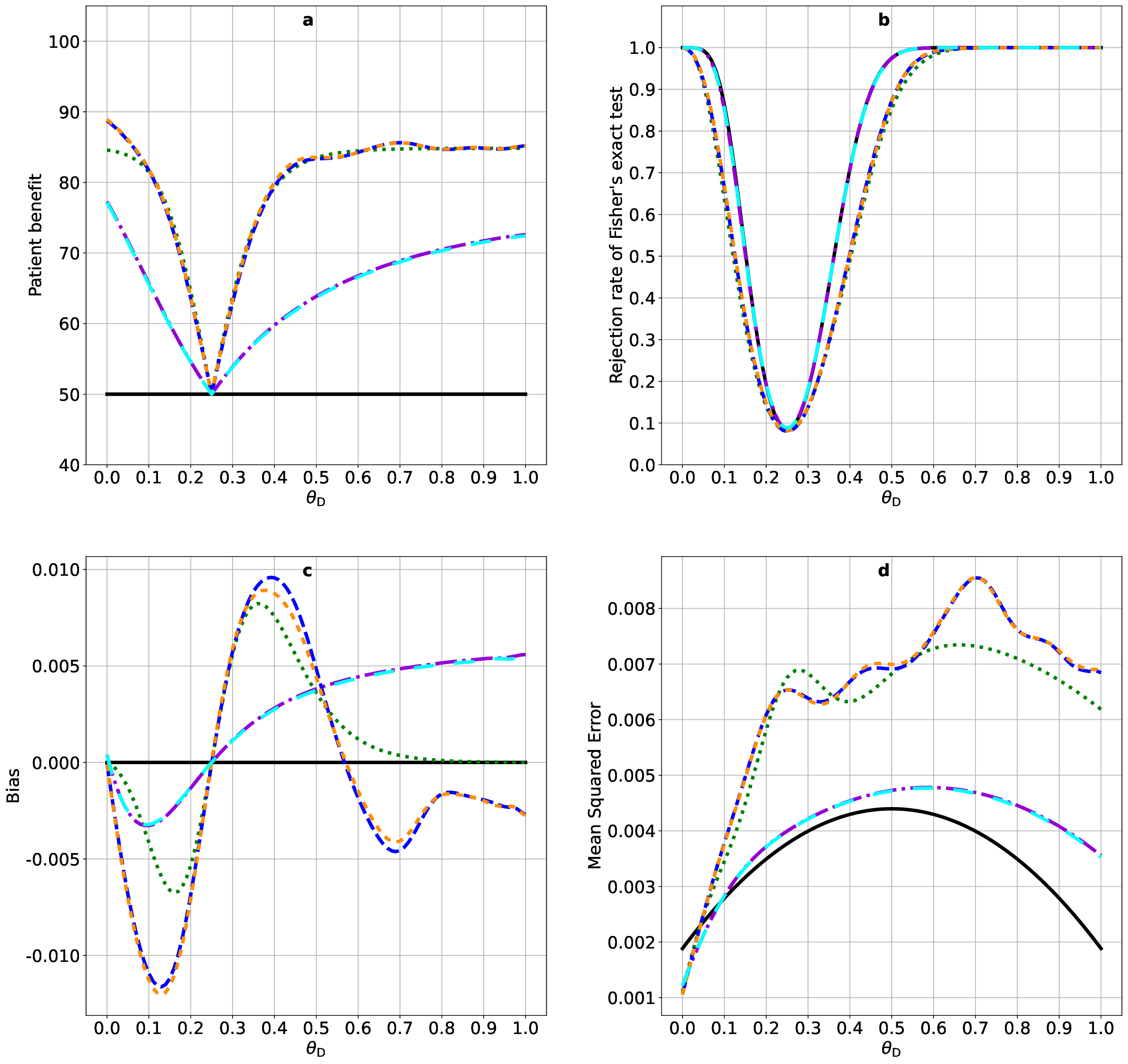}
		
	\caption{Patient benefit (subfigure~a), rejection rate (subfigure~b),  bias (subfigure~c), and mean squared error (subfigure~d) vs.~$\theta_\D$ for~$\theta_\C=0.25$,~$n=200$, and RA~procedures ER (solid), CMDP-E2 
with~$p=0.95$ (dashed, moderate), CMDP-E2 
with~$p=1.00$ (dashed, short), CRDP (dotted),  \mbox{CMDP-E1} 
with~$p=0.95$ (dash dot), CMDP-E2 
with~$p=1.00$ (dashed, long)}\label{results_N200_25_estimation}
\end{figure}

 \begin{figure}[htb]
		\includegraphics[width=\linewidth]{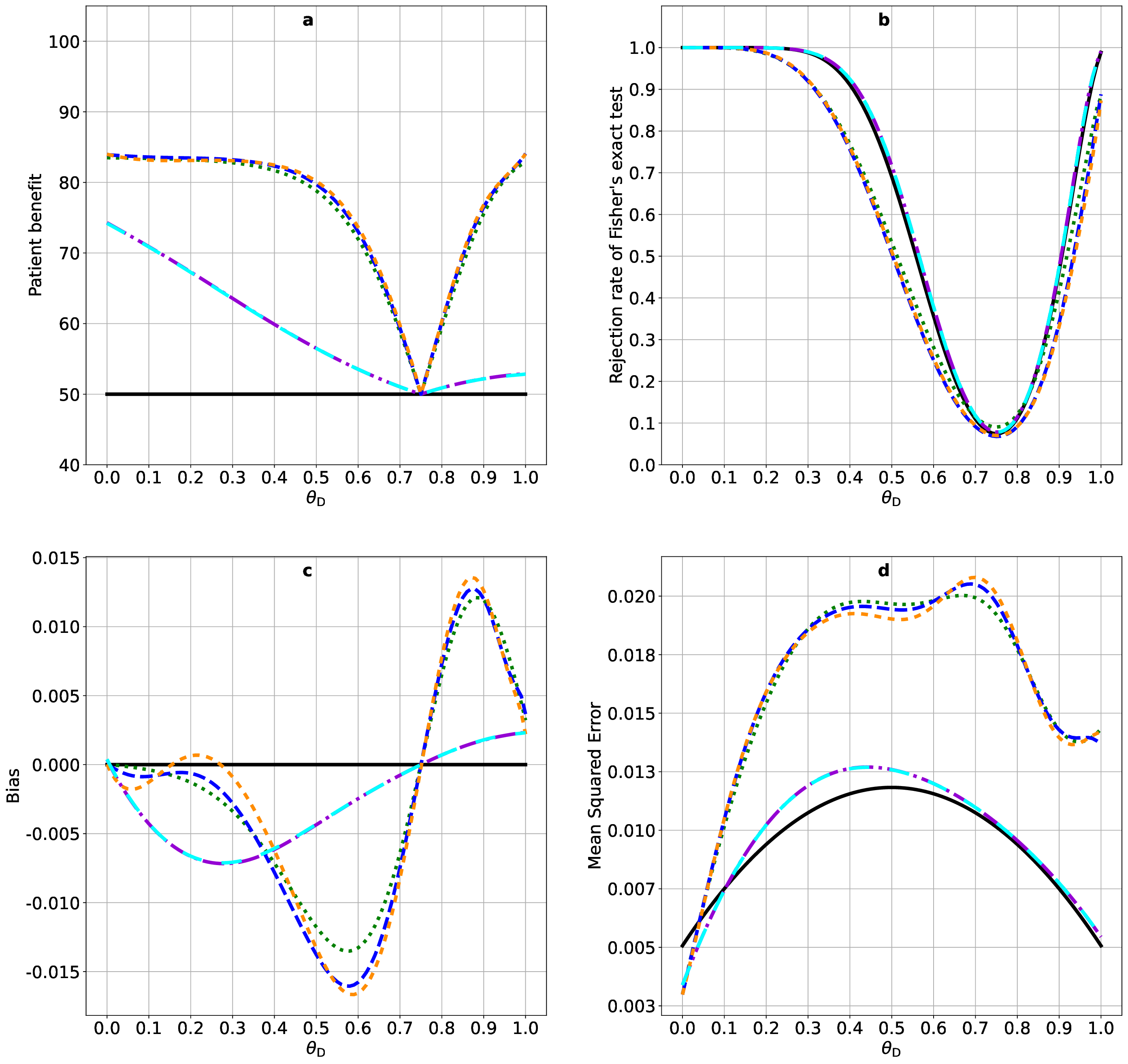}

	\caption{Patient benefit (subfigure~a), rejection rate (subfigure~b),  bias (subfigure~c), and mean squared error (subfigure~d) vs.~$\theta_\D$ for~$\theta_\C=0.75$,~$n=75$, and RA~procedures ER (solid), CMDP-E2 
with~$p=0.95$ (dashed, moderate), CMDP-E2 
with~$p=1.00$ (dashed, short), CRDP (dotted),  \mbox{CMDP-E1} 
with~$p=0.95$ (dash dot), CMDP-E2 
with~$p=1.00$ (dashed, long)}\label{results_N75_75_estimation}
\end{figure}

 \begin{figure}[htb]
		\includegraphics[width=\linewidth]{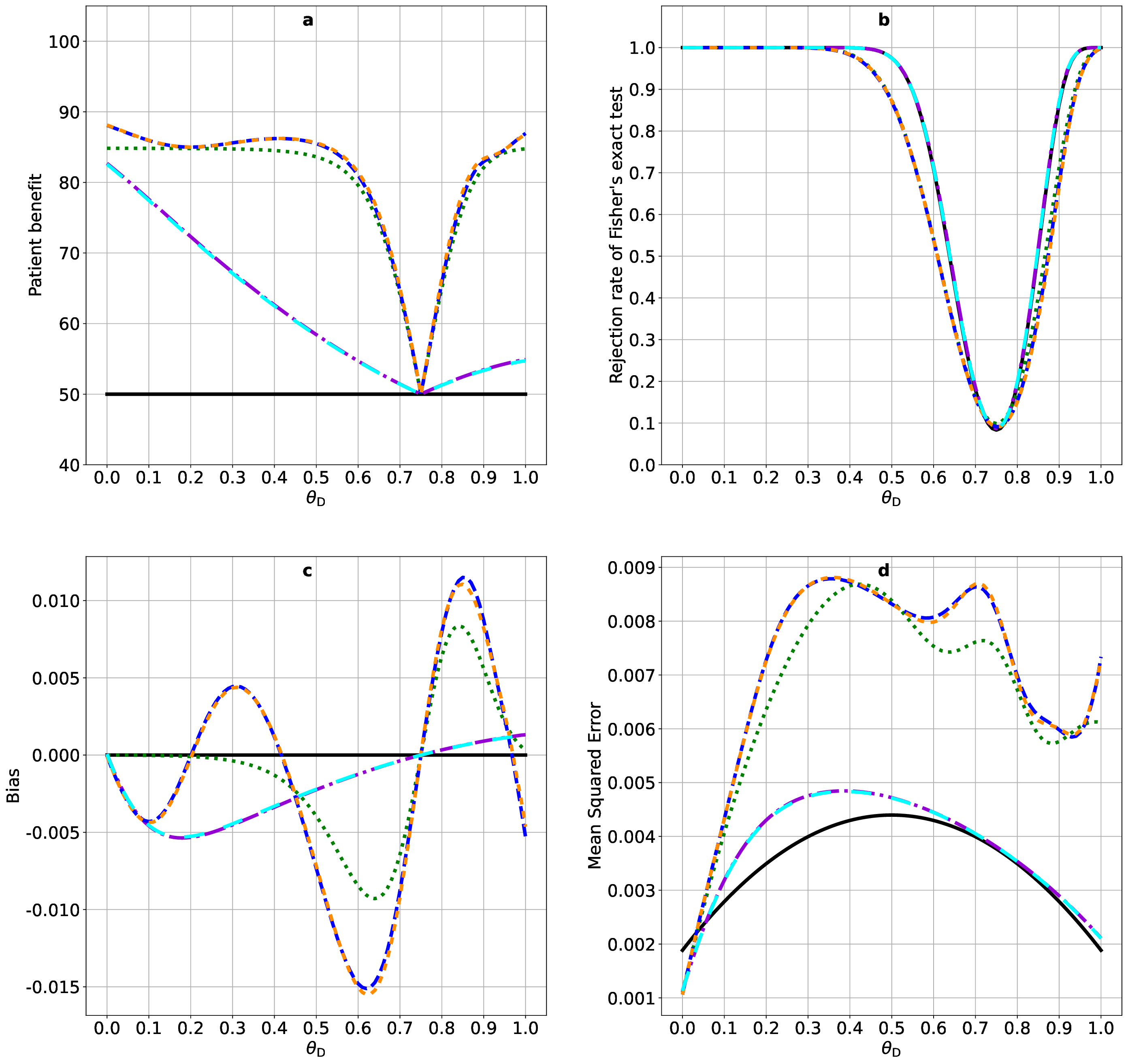}
		
	\caption{Patient benefit (subfigure~a), rejection rate (subfigure~b),  bias (subfigure~c), and mean squared error (subfigure~d) vs.~$\theta_\D$ for~$\theta_\C=0.75$,~$n=200$, and RA~procedures ER (solid), CMDP-E2 
with~$p=0.95$ (dashed, moderate), CMDP-E2 
with~$p=1.00$ (dashed, short), CRDP (dotted),  \mbox{CMDP-E1} 
with~$p=0.95$ (dash dot), CMDP-E2 
with~$p=1.00$ (dashed, long)}\label{results_N200_75_estimation}
\end{figure}

\FloatBarrier

\end{document}